\documentclass[11pt]{article}

\usepackage[utf8]{inputenc}
\usepackage[T1]{fontenc}    
\usepackage{hyperref}       
\usepackage{url}            
\usepackage{booktabs}      
\usepackage{amsfonts}       
\usepackage{nicefrac}       
\usepackage{microtype}     
\usepackage{xcolor}         
\usepackage{tabularx}
\usepackage{mathtools}
\usepackage{geometry}
\usepackage[utf8]{inputenc}
\usepackage{cite}
\usepackage{babel}
\usepackage{amsthm}
\usepackage{amsmath}
\usepackage{makecell}
\usepackage{amsfonts}
\usepackage{mathabx}
\usepackage{mathtools}
\usepackage{booktabs}
\usepackage{amsmath}
\usepackage{float}
\usepackage{thmtools,thm-restate}
\usepackage{bm}
\usepackage{enumitem}
\usepackage{bbm}
\usepackage[linesnumbered,ruled,vlined]{algorithm2e}
\usepackage{mathtools}
\usepackage{hyperref}
\usepackage{subcaption}
\usepackage[normalem]{ulem}
\useunder{\uline}{\ul}{}
\usepackage[T1]{fontenc}
\usepackage{xcolor}
\usepackage{textcomp}
\usepackage{graphicx}
\usepackage{amssymb}
\usepackage{pdfpages}
\usepackage{booktabs}
\usepackage{cleveref}

{\par\egroup\vskip 0.25ex}
\newlength\aftertitskip     \newlength\beforetitskip
\newlength\interauthorskip  \newlength\aftermaketitskip

\DeclareMathOperator*{\tr}{\text{tr}}

\newtheorem{lemma}{Lemma}

\newtheorem{remark}{Remark}
\newtheorem{definition}{Definition}

\SetKw{Break}{break}
\SetKw{Return}{return}
\SetCommentSty{mycommfont}
\newtheorem{theorem}{Theorem}

\renewcommand\b{{\bf b}} 
 
\newcommand{\defeq}{:=}
\newcommand{\dmin}{d_{\min}}
\newcommand{\dmax}{d_{\max}}

\newcommand{\normFull}[1]{\left\|#1\right\|}
\newcommand{\norm}[1]{\|#1\|}

\allowdisplaybreaks

\newcommand{\pars}[1]{\left( #1 \right)}
\newcommand{\curs}[1]{\left\{ #1 \right\}}
\newcommand{\sqrs}[1]{\left[ #1 \right]}
\newcommand{\abs}[1]{\left| #1 \right|}
\newcommand{\Norm}[1]{\left\Vert #1 \right\Vert}

\newcommand{\R}{\mathbb{R}}
\newcommand{\Z}{\mathbf{Z}}
\newcommand{\I}{\mathbf{I}}
\newcommand{\A}{\mathbf{A}}
\newcommand{\B}{\mathbf{B}}

\newcommand{\Smat}{\mathbf{S}}
\newcommand{\As}{\widebar{\mathbf{A}}}
\newcommand{\Bs}{\widebar{\mathbf{B}}}
\newcommand{\rad}{\xi}
\newcommand{\Ntilde}{\widetilde{\mathbf{N}}}

\newcommand{\normgen}[2]{\left\|#1 \right\|_{#2} }

\newcommand{\trace}[1]{\mathrm{tr}(#1)}

\newcommand{\Prob}[1]{\mathbb{P} \left[ #1 \right]}

\newcommand{\one}{\mathbbm{1}}

\newcommand{\Rho}[1]{\rho\left( #1 \right)}

\newcommand{\D}{\mathbf{D}}

\newcommand{\Q}{\mathbf{Q}}
\newcommand{\N}{\mathbf{N}}
\newcommand{\Ps}{\mathbf{P}}
\newcommand{\Y}{\mathbf{Y}}
\newcommand{\X}{\mathbf{X}}
\newcommand{\W}{\mathbf{W}}
\newcommand{\M}{\mathbf{M}}
\newcommand{\J}{\mathbf{J}}
\newcommand{\Mp}{\mathbf{M}'}
\newcommand{\Qp}{\mathbf{Q}'}

\newcommand{\Abar}{\widebar{\mathbf{A}}}
\newcommand{\Ldagger}{\mathbf{L}^\dagger}

\newcommand{\rtildeij}{\widetilde{r}_{i,j}}
\newcommand{\effres}{r}
\newcommand{\rGij}{{\effres}_G({i,j})}
\newcommand{\rGab}{{\effres}_{G}(a,b)}

\newcommand{\Gprime}{H}
\newcommand{\deltaij}{\vec{\delta}_{i,j}}
\newcommand{\deltaab}{\vec{\delta}_{a,b}}

\newcommand{\Abs}[1]{ \left\lvert #1 \right\rvert}

\newcommand{\Exp}[1]{\exp \left( #1 \right)}

\newcommand{\Otilde}[1]{\widetilde{O}( {#1} )}

\newcommand{\mlap}{\mathbf{L}}
\newcommand{\laplacian}[1]{\mlap_{#1}} 
\newcommand{\nnz}{\mathrm{nnz}} 
\newcommand{\nnzFull}[1]{\mathrm{nnz}\left({#1}\right)} 

\newcommand{\nsD}[2]{\mathrm{ns}_{#1}\left({#2}\right)}

\newcommand{\normalized}[1]{\mathbf{N}_{#1}} 
\newcommand{\normalizedlazy}[1]{{\mathbf{N}_\ell}_{#1}}
\newcommand{\diagonal}[1]{{\mathbf{D}}_{#1}}
\newcommand{\adjacency}[1]{{\mathbf{A}}_{#1}}
\newcommand{\adjacencylazy}[1]{{\mathbf{A}_\ell}_{#1}} 

\newcommand{\nullspace}[1]{\mathrm{ker}\left(#1\right)}
\newcommand{\Spectrum}[1]{\lambda\left(#1\right)}

\newcommand{\lambdamin}{\lambda_{\mathrm{min}}}
\newcommand{\pcondOf}[1]{\widebar{\kappa}\left({#1}\right)}

\newcommand{\timeSketch}{\mathcal{T}_{\mathrm{s}}}
\newcommand{\timeQuery}{\mathcal{T}_{\mathrm{q}}}

\title{Towards Optimal Effective Resistance Estimation}
\author{Rajat Dwaraknath \\ Stanford University \\ \texttt{rajatvd@stanford.edu} \and Ishani Karmarkar \\ Stanford University \\ \texttt{ishanik@stanford.edu} \and Aaron Sidford \\ Stanford University \\ \texttt{sidford@stanford.edu}}

\date{}
\begin{document}
\maketitle

\begin{abstract}%
We provide new algorithms and conditional hardness for the problem of estimating effective resistances in $n$-node, $m$-edge, undirected, expander graphs. We provide an $\widetilde{O}(m\epsilon^{-1})$-time algorithm that produces with high probability, an $\widetilde{O}(n\epsilon^{-1})$-bit sketch from which the effective resistance between any pair of nodes can be estimated, to $(1 \pm \epsilon)$-multiplicative accuracy, in $\widetilde{O}(1)$-time. Consequently, we obtain an $\widetilde{O}(m\epsilon^{-1})$-time algorithm for estimating the effective resistance of all edges in such graphs, improving (for sparse graphs) on the previous fastest runtimes of $\widetilde{O}(m\epsilon^{-3/2})$ \cite{chu2018graph} and $\widetilde{O}(n^2\epsilon^{-1})$ \cite{JS2018} for general graphs and $\widetilde{O}(m + n\epsilon^{-2})$ for expanders \cite{li2022new}. 
We complement this result by showing a conditional lower bound that a broad set of algorithms for computing such estimates of the effective resistances between all pairs of nodes require $\widetilde{\Omega}(n^2 \epsilon^{-1/2})$-time, improving upon the previous best such lower bound of $\widetilde{\Omega}(n^2 \epsilon^{-1/13})$ \cite{musco2017spectrum}. Further, we leverage the tools underlying these results to obtain improved algorithms and conditional hardness for more general problems of sketching the pseudoinverse of positive semidefinite matrices and estimating functions of their eigenvalues.
\end{abstract}

\tableofcontents

\section{Introduction}\label{section:intro}
In a weighted, undirected graph $G$ the \emph{effective resistance} between a pair of vertices $a$ and $b$, denoted $\rGab$, is the energy of a unit of electric current sent from $a$ to $b$ in the natural resistor network induced by $G$. Effective resistances arise for a broad set of graph processing tasks and have multiple equivalent definitions. For example, $\rGab$ is proportional to the expected roundtrip commute time between $a$ and $b$ in the natural random walk induced on the graph and when $\{a,b\}$ is an edge in the graph, it is proportional to the probability that the edge is in a random spanning tree.

Effective resistances are also a metric on the vertices \cite{Madry_2014, Devriendt_2022} and are a key measure of proximity between vertex pairs. Correspondingly, effective resistances can arise in a variety of data analysis tasks. For example, effective resistances have been used in social network analysis for measuring edge centrality in social networks \cite{kirchoff} as well as for measuring chemical distances \cite{Klein_1993}.  

Effective resistances have a broad range of algorithmic implications. Sampling edges of a graph using effective resistance is known to efficiently produce cut and spectral sparsifiers (sparse graphs which approximately preserve cuts, random walk properties, and more) \cite{spielman2009graph, lee2015constructing, durfee2019fully}. Effective resistance-based graph sparsifiers have also been applied to develop fast graph attention neural networks \cite{srinivasa2020fast}, to design graph convolutional neural networks for action recognition \cite{AHMAD2021389}, to sample from Gaussian graphical models \cite{pmlr-v40-Cheng15}, and beyond \cite{pmlr-v80-calandriello18a, Mercier_2022}. Effective resistances have been used in algorithms for maximum flow problems, \cite{christiano2010electrical, madry2013navigating, madry2016computing, van2022faster, van2021minimum, lee2014path, van2022faster}, sampling random spanning trees \cite{durfee2017sampling,schild2018, madry2014fast}, and graph partitioning \cite{alev2017graph, zhao2019}. More recently, effective resistances have also been used to analyze the problem of oversquashing in GNNs and in designing algorithms to alleviate oversquashing \cite{digiovanni2023oversquashing, Banerjee2022, black2023understanding} and have been applied to increase expressivity when incorporated as edge features in certain GNNs \cite{velingker2022affinityaware}.

\paragraph{Algorithms.}
Given the broad utility of effective resistances, there have been many methods for estimating and  approximately compressing them \cite{JS2018, li2022new, chu2018graph, spielman2009graph, goranci2018dynamic}. In this paper, our main focus is the following effective resistance estimation problem. (We use $x \approx_\epsilon y$ as shorthand for $(1 - \epsilon) y \leq x \leq (1 + \epsilon) y$ and assume all edge weights in graphs are positive. See Section~\ref{section:preliminaries} for notation more broadly.)

\begin{definition}[Effective Resistance Estimation Problem]
\label{def:effres_est}
In the \emph{effective resistance estimation problem} we are given an undirected, weighted graph $G = (V,E,w)$, a set of vertex pairs $S \subseteq V \times V$, and $\epsilon \in (0, 1)$ and must output $\tilde{r} \in \R^{S}$ such that whp., $\tilde{r}_{(a,b)} \approx_\epsilon \rGab$ for all $(a,b) \in S$. 
\end{definition}

The state-of-the-art runtimes for solving the effective resistance estimation problem on $n$-node, $m$-edge graphs are given in Table~\ref{table:estimate}. To contextualize these results, consider the special case of estimating the effective resistance of a graph's edges, i.e., when $S = E$. This special case appears in many of the aforementioned applications, e.g., \cite{christiano2010electrical, schild2018, madry2013navigating, madry2016computing,durfee2017sampling}. The state-of-the-art runtimes for this problem are an $\Otilde{n^2 \epsilon^{-1}}$ time algorithm \cite{JS2018} and an $\Otilde{m \epsilon^{-1.5}}$ algorithm \cite{chu2018graph}. A major open problem is whether improved runtimes, e.g., $\Otilde{m \epsilon^{-1}}$ (which would subsume prior work), are attainable.

One of the main results of this paper is resolving this open problem in the case of well-connected graphs, i.e., expanders.
Formally, we parameterize our bounds in terms of the graph's expansion $\pcondOf{G}$, defined in Section~\ref{sec:preliminaries} and provide a number of results for \emph{expander graphs}, i.e., when $\pcondOf{G} = \Otilde{1}$. In particular, we provide an $\Otilde{m \epsilon^{-1} \pcondOf{G}}$ time algorithm for effective resistance estimation when $S = E$. Previously, the only non-trivial improvement in this setting was an independently obtained runtime of  $\Otilde{m + n\epsilon^{-2} (\pcondOf{G})^{3}}$.\footnote{While our algorithms for the effective resistance estimation problem (Definition~\ref{def:effres_est}) were obtained indpendently, our writing and discussion of effective resistance sketch algorithms (Definition~\ref{def:eff_res_sketch}) was informed by their paper. We provide a more complete comparison in Table~\ref{table:sketch}.} Expanders are a non-trivial, previously studied special case that is often the first step or a key component for developing more general algorithms \cite{dinitz2015resistance}. 

Interestingly, we obtain our main result by providing new effective resistance sketch algorithms.

\begin{definition}[Effective Resistance Sketch]
\label{def:eff_res_sketch}

We call a randomized algorithm an \emph{$(\timeSketch, \timeQuery, s)$-effective resistance sketch algorithm} if given an input $n$-node, $m$-edge  undirected, weighted graph $G = (V,E,w)$ and $\epsilon \in (0, 1)$ in time $O(\timeSketch(G,\epsilon))$ it creates a binary string of length $O(s(G,\epsilon))$ from which when queried with any $a,b\in V$, it outputs $\tilde{r}_{a,b} \approx_\epsilon \rGab$ whp.\ in time $O(\timeQuery(G,\epsilon))$. 
\end{definition}

Effective resistance sketching algorithms immediately imply algorithms for the effective resistance estimation problem. We obtain our result by obtaining an $(\Otilde{n \epsilon^{-1}}, \Otilde{m\epsilon^{-1}})$-effective resistance sketch algorithm for expanders (see Section~\ref{section:intro_results} for a comparison to prior work). 

\paragraph{Lower Bounds.} Given the central role of the effective resistance estimation and the challenging open-problem of determining its complexity, previous work has sought complexity theoretic lower bounds for the problem. \cite{musco2017spectrum} showed a conditional lower bound of $\Omega(n^2\epsilon^{-1/13})$ for the problem by showing that an algorithm that computes effective resistances in $O(n^2\epsilon^{-1/13 + \delta})$ time for some $\delta > 0$ could be used to obtain a subcubic algorithm for triangle detection in undirected graphs.

This \emph{triangle detection problem} is the problem of determining whether an undirected, unweighted graph contains a triplet of edges $(\{a, b\}, \{b, c\}, \{c, a\})$. Currently, the only known subcubic algorithms for the triangle detection problem leverage fast matrix multiplication (FMM) and therefore their practical utility (in the worst case) is questionable. Meanwhile, there are no known deterministic or randomized algorithms for the triangle detection problem that do not use FMM. \cite{william_williams2018} showed that any algorithm which solves this triangle detection problem in subcubic time could be used to obtain a subcubic algorithm for Boolean matrix multiplication (BMM) and additional problems which currently are only known to be solvable subcubically with FMM. 
Consequently, subcubic triangle detection is a common hardness assumption used to illustrate barriers towards improving non-FMM based methods, e.g., the effective resistance estimation algorithms of this paper.
 
In this paper we take a key step towards closing the gap between the best known running times for effective resistance estimation and lower bounds by improving the conditional lower bound of $\Omega(n^2\epsilon^{-1/13})$ to $\widetilde{\Omega}(n^2 \epsilon^{-1/2})$ for randomized algorithms. We show this lower bound holds \emph{even for expander graphs}, and hence our effective resistance estimation algorithm (as well as \cite{JS2018}) are optimal up to an $\epsilon^{-1/2}$-factor among non-FMM based algorithms, barring a major breakthrough in BMM.

\paragraph{Broader Linear Algebraic Tools.}
The effective resistance between vertex $a$ and vertex $b$ in a graph $G$ has a natural linear algebraic formulation. For all $a,b \in V$ it is known that $r_G(a,b) = \deltaab \laplacian{G}^\dagger \deltaab$, where $\laplacian{G} \in \R^{V \times V}$ is a natural matrix associated with $G$ known as the \emph{Laplacian matrix} and $\deltaab = e_a - e_b$ (see Section~\ref{section:preliminaries} for notation). Thus, sketching effective resistances can be viewed as problems of preserving information about subsets of entries of the pseudoinverse of a Laplacian. 

To obtain our algorithms and lower bounds we develop tools that apply to related problems for more general (not-necessarily Laplacian) matrices. In terms of algorithms, we show our techniques lead to algorithms and data structures for computing certain quadratic forms involving well-conditioned SDD and PSD matrices. In terms of hardness, we show our techniques improve triangle detection hardness bounds for estimating various properties of the singular values of a matrix. 

\paragraph{Paper Organization.} 

In the remainder of the introduction we provide a more precise statement and comparison of our results in Section~\ref{section:intro_results}. 
In the remainder of the paper we cover preliminaries in Section~\ref{sec:preliminaries}, present upper bounds in Section~\ref{section:upper_bounds}, and present lower bounds in Section~\ref{sec:lower_bounds}. 

\subsection{Our Results}\label{section:intro_results}

\paragraph{Algorithms.} { Here we outline our main algorithmic results pertaining to effective resistance sketching and estimation, and in Section~\ref{section:upper_bounds} we describe some extensions of our work to broader linear algebraic problems involving SDD and PSD matrices. Our main algorithmic result is a new efficient effective resistance sketch for \textit{expanders}, a term which is used to refer to graphs with $\widetilde{\Omega}(1)$-expansion. 

\begin{definition}[Expander]\label{def:expander} For $\alpha > 0$, we say that a graph $G = (V, E, w)$ has $\alpha$-expansion if $\alpha \leq \phi(G)$, where $\phi(G)$ denotes the conductance of $G$ and is defined as
\begin{align*}
    \phi(G) := \min_{S \subseteq V, S\notin \{0, V\}} \frac{\sum_{\{u, v\} : u \in S, v \in V\setminus{S}} w_{u, v}}{\min\left(\sum_{u \in S} d_u, \sum_{v \in V\setminus{S}} d_u\right)}, \text{ where } d_u := \sum_{\{u, v\} \in E} w_u.
\end{align*}
We say a graph is an expander if it has $\widetilde{\Omega}(1)$-expansion. 
\end{definition}

\begin{theorem}\label{thm:ub_main_informal} There is an $(\Otilde{m\epsilon^{-1}}, \Otilde{1}, \Otilde{n\epsilon^{-1}})$ effective resistance sketch algorithm for expanders.
\end{theorem}
}

Table~\ref{table:sketch} summarizes and compares our Theorem~\ref{thm:ub_main_informal} to previous work on effective resistance sketches, including naive algorithms to explicitly compute the pseudoinverse of the Laplacian of $G$, which can be computed in $O(n^\omega)$ time using FMM or $\Otilde{mn}$ time using a Laplacian system solver (labeled Solver).\footnote{$\omega \leq n^{2.37188}$ \cite{duan2023faster} denotes the fast matrix multiplication constant. } A $(\timeSketch, \timeQuery, s)$
effective-resistance sketch algorithm implies an $O(\timeSketch + |S| \timeQuery )$ algorithm for the effective resistance estimation problem. Hence, Theorem~\ref{thm:ub_main_informal} implies the following. 
\begin{theorem}[Effective Resistance Estimation on Expanders] \label{thm:ub_main_informal2} There is an $\Otilde{m\epsilon^{-1} + |S|}$ time algorithm which solves the effective resistance estimation problem for expanders whp.
\end{theorem}

Effective resistance sketches are a common approach to solving the effective resistance estimation problem; but there are also approaches to the problem that do not explicitly construct effective resistance sketches. Table~\ref{table:estimate} summarizes prior work on effective resistance estimation more broadly.

There has been a long line of research on the problem of computing sketches and sparsifiers of graph Laplacians \cite{andoni2015sketching, spielman2009graph, JS2018, chu2018graph} (i.e., computing a sparse graph $G'$ such that quadratic forms in the Laplacian of $G'$ approximate quadratic forms in the Laplacian of $G$). Building on this work, \cite{chu2018graph} showed there is an algorithm which processes a graph $G$ on $n$ nodes and $m$ edges in $O(m^{1 + o(1)})$ time and produces a sparse sketch graph $H$ with only $\Otilde{n\epsilon^{-1}}$ edges such that $r_G(a,b)\approx r_H(a, b)$ for all $a, b \in V$. Consequently, any algorithm which runs in $\Otilde{m\epsilon^{-c}}$ on expanders can be improved to run in $\Otilde{m^{1+o(1)} + n\epsilon^{-(c+1)}}$ time on expanders simply by running the algorithm on $H$ instead of $G$.

\begin{table}[ht]
\hrule
\centering
\begin{minipage}[t]{.45\linewidth}
    \centering
    \caption{Sketch}
    \label{table:sketch}
    \begin{tabular}{c c c c c c }
    \toprule
    Method & $\mathcal{T}_p$ & $\mathcal{T}_q$ & $s$ \\
    \midrule
    FMM & $ n^\omega $ & $1$ & $n^2$ \\
    Solver & $ nm $ & $1$ & $n^2$ \\
    \cite{spielman2009graph} & $m\epsilon^{-2}$ & $\epsilon^{-2}$ & $n\epsilon^{-2}$ \\
    \cite{JS2018} & $n^2\epsilon^{-1}$ & 1 & $n\epsilon^{-1}$ \\
    \cite{li2022new} & $m + n \epsilon^{-2}$ & $1$ & $n\epsilon^{-1}$ \\
    Ours & $m \epsilon^{-1}$ & $1$ & $n\epsilon^{-1}$ \\
    \end{tabular}
\end{minipage}%
\begin{minipage}[t]{.55\linewidth}
    \centering
    \caption{Effective Resistance Estimation}
    \label{table:estimate}
    \begin{tabular}{c c c  }
    \toprule
     Method & Runtime & \makecell{\makecell{Restriction}} \\
    \midrule
    FMM & $n^\omega + \abs{S}$ & None \\
    Solver & $n m + \abs{S}$ & None \\
    \cite{spielman2009graph} & $n^2\epsilon^{-2}$ & $S = V \times V$ \\
    \cite{JS2018} & $n^2\epsilon^{-1}$ & $S = V \times V$\\
    \cite{durfee2017sampling} & $m + (n + \abs{S})\epsilon^{-2}$ & None \\
    \cite{HuanSchild2018} & $m + \abs{S}\epsilon^{-2}$ & None \\
    \cite{chu2018graph} & $m\epsilon^{-1.5}$ & $S = E$ \\
    \cite{li2022new} & $m + n\epsilon^{-2} + \abs{S}$ & None \\
    Ours & $m\epsilon^{-1} + \abs{S}$ & None \\
    \end{tabular} \\
\end{minipage}
\\[1ex]
\hrule 
\vspace*{.2cm}
\parbox{\linewidth}{\textbf{Prior work on Effective Resistance Sketches (Table~\ref{table:sketch}) and Estimation (Table~\ref{table:estimate}).} Time and space complexities are reported for $n$-node, $m$-edge expanders
up to $\widetilde{O}(\cdot)$. Our results and \cite{li2022new} apply only to expanders; however, the remaining works apply to general graphs. As discussed, when $m^{1 + o(1)} + n \epsilon^{-(c+1)} = o(m\epsilon^{-c})$, any runtime dependence on $m\epsilon^{-c}$ in the table can be improved to a dependence on $m^{1 + o(1)} + n \epsilon^{-(c+1)}$.
} 
\vspace{-1em}
\end{table}

\paragraph{Lower Bounds} 
For the effective resistance estimation problem, \cite{musco2017spectrum} showed that any algorithm which solves the effective resistance estimation problem for $S = V \times V$ in $\Otilde{n^2\epsilon^{-1/13 + \delta}}$ for some $\delta > 0$, would imply a combinatorial subcubic deterministic algorithm which detects a triangle in an $n$-node undirected unweighted graph. We improve on their result, as follows. 

\begin{theorem}\label{thm:reff_lb_main_result} Given an algorithm which solves the effective resistance estimation problem for $S = V\times V$ on graphs with $\widetilde{\Omega}(1)$-expansion in $\Otilde{n^2\epsilon^{-1/2 + \delta}}$ time for $\delta > 0$, we can produce a randomized algorithm that solves the triangle detection problem on an $n$-node graph in $\Otilde{n^{3-2\delta}}$ time whp. 

\end{theorem}
Theorem~\ref{thm:reff_lb_main_result} implies an $\widetilde{\Omega}(n^2\epsilon^{-1/2})$ randomized conditional lower bound for the problem of estimating effective resistances of all pairs of nodes in an undirected unweighted expander graph, while \cite{musco2017spectrum} shows only an $\Omega(n^2\epsilon^{-1/13})$ lower bound. 

In addition to conditional lower bounds for effective resistance estimation, we also improve on existing conditional lower bounds for the problem of estimating spectral sums that we define below. The definition is inspired by Theorem 15 of \cite{musco2017spectrum}.

\begin{definition}[Spectral Sum]\label{def:spectral_sum} For $f:\R^+ \to \R^+$ and $\A \in \R^{n \times n}$ with singular values $\sigma_1(\A) \leq \sigma_2(\A) \leq \cdots \leq \sigma_n(\A)$, we define the spectral sum $\mathcal{S}_f : \R^{n \times n} \to \R^+$ as $\mathcal{S}_f (\A) := \sum_{i=1}^n f (\sigma_i(\A))$.
\end{definition}

\cite{musco2017spectrum} showed that for several spectral sums $\mathcal{S}_f$, any algorithm that outputs $Y \approx_\epsilon \mathcal{S}_f(\A)$ in $O(n^\gamma\epsilon^{-c})$ time for $\gamma \geq 2$ on an $n \times n$ PSD matrix would imply an $O(n^{\gamma + \alpha c})$ time algorithm which solves the triangle detection problem, where the scaling $\alpha$ varies depending on the specific $\mathcal{S}_f$ (see Table~\ref{table:spectral_sum_hardness}). We build on their results to show improved randomized conditional lower bounds for several spectral sum estimation problems, as presented in Theorem~\ref{thm:our_results_spectralsums_lowerbound} below. 

\begin{restatable}{theorem}{thmOurResultsSpectralSumsLowerBound}\label{thm:our_results_spectralsums_lowerbound}
Given an algorithm which on input $\B \in \R^{n \times n}$ outputs a spectral sum estimate $Y \approx_\epsilon \mathcal{S}_f(\B)$ in $O\pars{n^\gamma \epsilon^{-c}}$ time with $\gamma \geq 2$ for the spectral sums in Table~\ref{table:spectral_sum_hardness}, we can produce an algorithm that can detect a triangle in an n-node graph whp.\ in $\widetilde{O}\pars{n^{\gamma + \alpha c}}$ time, where $\alpha$ is a scaling that depends on properties of the function $f$ (see Table~\ref{table:spectral_sum_hardness} for values of $\alpha$ for several spectral sums.)
\end{restatable}

\newcolumntype{Y}{>{\centering\arraybackslash}X}
\newcolumntype{Z}{>{\centering\arraybackslash}XX} 

\begin{table}[ht]
\centering
\begin{tabularx}{\textwidth}{c Y Y Y Y}
\toprule
& \multicolumn{2}{ c }{\cite{musco2017spectrum}} & \multicolumn{2}{ c }{This Paper} \\
\cmidrule(lr){2-3} \cmidrule(lr){4-5}
Spectral Sum & TD Runtime & Lower Bound & TD Runtime & Lower Bound \\
\midrule
Schatten 3-norm & $n^{\gamma + 4c}$ & $n^2 \epsilon^{-1/4}$ & $n^{\gamma + 5c/2}$ & $n^2 \epsilon^{-2/5}$ \\
Schatten p-norm, $p\neq 1,2$, &  $n^{\gamma + 13c}$ & $n^2 \epsilon^{-1/13}$ & $n^{\gamma + 10c}$ & $n^2 \epsilon^{-1/10}$ \\
SVD Entropy & $n^{\gamma + 6c}$ & $n^2 \epsilon^{-1/6}$ & $n^{\gamma + 5c}$ & $n^2 \epsilon^{-1/5}$ \\
Log Determinant &$n^{\gamma + 6c}$ & $n^2 \epsilon^{-1/6}$ & $n^{\gamma + 5c}$ & $n^2 \epsilon^{-1/5}$ \\
Trace of Exponential & $n^{\gamma + 13c}$ & $n^2 \epsilon^{-1/13}$ & $n^{\gamma + 10c}$ & $n^2 \epsilon^{-1/10}$ \\
\bottomrule \\ 
\end{tabularx} \\
\caption{\textbf{Runtimes for the triangle detection (TD) problem in an $n$-node graph using algorithms that produce $(1\pm\epsilon)$ multiplicative approximations to various spectral sums in $O(n^\gamma \epsilon^{-c})$ time}. The second columns contain the best achievable runtimes for $\gamma = 2$ that do not use FMM, barring a breakthrough in subcubic triangle detection. Runtimes are reported up to $\Otilde{\cdot}$.} 
\label{table:spectral_sum_hardness}
\end{table}

\subsection{Additional Related Work}\label{sec:additional_related_work}

Here we briefly discuss additional work related to the effective resistance estimation problem and provide a more detailed comparison of our results to \cite{li2022new}. 

\paragraph{Dynamic effective resistance estimation.}{ Effective resistance estimation and sketching are part of a broader family of previously studied problems involving graph compression and effective resistance estimation. For example, there is a related line of work on dynamically maintaining effective resistance  estimates in dynamic graphs, e.g., \cite{durfee2017sampling, goranci2018dynamic}, which in turn is related to problems of dynamically maintaining electric flows in graphs, e.g., \cite{maxflow,BrandGJLLPS22}. Whether our techniques have ramifications for these related problems is an interesting question for future work. 
}

\paragraph{Fine-grained complexity analysis.}{Our effective resistance estimation lower bounds fall under a broader topic of fine grained complexity analysis, i.e., the problem of characterizing the optimal complexity of problems which are known to have polynomial time solutions. Here, we provide references to a few examples. \cite{william_williams2018} showed subcubic equivalences between the problem of triangle detection, Boolean matrix multiplication, and several other graphical problems. As discussed, \cite{musco2017spectrum} utilize the results of \cite{william_williams2018} to obtain conditional lower bounds on several spectral sum approximation problems -- many of which we also study in this paper. Similarly, \cite{garciamarco2017complexity, Mor73, shpilka2002lower} provided several conditional complexity lower bounds for linear algebraic problems, conditional on the use of particular computational models. \cite{williams2022algorithms} and \cite{boroujeni2019subcubic} also provide fine grained lower bounds for the fault replacement paths problem, problems on graph centrality measures, and complementary problems. Making connections between our techniques for our effective resistance estimation lower bounds and these prior works in fine-grained complexity analysis is an interesting open problem. 

\paragraph{Comparison to \cite{li2022new}.}{ Effective reisistance sketching and estimation for expanders was previously studied in \cite{li2022new}. \cite{li2022new} produces an $(\Otilde{m + n\epsilon^{-2}}, \Otilde{1}, \Otilde{n\epsilon^{-1}})$} effective resistance sketch for expanders. Our work provides a different, independently obtained runtime for effective resistance estimation by producing an $(\Otilde{m\epsilon^{-1}}, \Otilde{1}, \Otilde{n\epsilon^{-1}})$} effective resistance sketch for expanders. Additionally, our work can be applied to produce an $(\Otilde{m^{1+o(1)} + n\epsilon^{-2}}, \Otilde{1}, \Otilde{n\epsilon^{-1}})$ effective resistance sketch, by running on a sparse graphical sketch, such as that guaranteed by \cite{chu2018graph} (see Section~\ref{section:intro_results}). Consequently, our results match those of \cite{li2022new} up to an $m^{o(1)}$ factor, and improve for sufficiently high accuracy on sparse graphs.  

Given an expander $G$, \cite{li2022new} considers storing $\Otilde{\epsilon^{-1}}$-sparse approximations to the columns of $\laplacian{G}^\dagger$, which would clearly be sufficient for querying effective resistances of $G$ in $\Otilde{1}$ time. However, it is unclear whether the columns of $\laplacian{G}^\dagger$ have small $\ell_1$ norm, and consequently, it is unclear how to obtain these sparse approximations. Consequently, their algorithm instead estimates the following vector $\sigma_u$ for each $u \in V$,
\begin{align*}
    \sigma_u = \frac{1}{2} \sum_{t=0}^\infty \left(\left(\frac{1}{2} \I + \frac{1}{2} \adjacency{G}\diagonal{G}^{-1} \right)^t e_u - \pi\right), 
\end{align*}
where $\pi = \frac{\diagonal{\laplacian{G}} \one}{\one^\top \diagonal{\laplacian{G}} \one}$. They show that $\sigma_u$ is closely related to $\diagonal{G} \laplacian{G}^\dagger e_u - \pi$, and consequently access to $\sigma_u$ is sufficient for estimating effective resistances. Additionally, they use structural properties of expanders to show that each $\sigma_u$ must have small $\ell_1$ norm and that it can be computed efficiently by running lazy random walks on $G$ (i.e., the random walk which, at each step follows the natural random walk on $G$ with probability 1/2 and stays idle with probability 1/2). These key properties of $\sigma_u$ enable their result. 

Our approach is similar to that of \cite{li2022new} in that we also reformulate the effective resistance between two vertices as an inner product between two vectors whose $\ell_1$ norm we can bound; however, the vectors we consider are not the $\sigma_u$ vectors considered in \cite{li2022new}. Instead, we rewrite $\rGij$ as the inner product of $\diagonal{\laplacian{G}}^{-1} \deltaij$ with $\diagonal{\laplacian{G}}^{1/2}(\normalized{\laplacian{G}}/2)^{\dagger} \diagonal{\laplacian{G}}^{-1/2} \deltaij$. Similar to \cite{li2022new}, we then use similar underlying structural properties of expanders to argue that the $\ell_1$ norms of these vectors is not too large. However, instead of using random walks to estimate these vectors, we use sketching techniques (specifically, CountSketch) and Laplacian system solvers to estimate them, an idea which is inspired by the work of \cite{spielman2009graph}.  The differences in the specific effective resistance vectors we estimate 
and the different technique of estimating them is what leads to the difference in runtime between \cite{li2022new} and our own work. \cite{li2022new} also provide extensions of their effective resistance sketch techniques to well-conditioned SDD matrices, which we also obtain in our generalized Theorem~\ref{theorem:theoremSDDMain}; our lower bounds on the SDD effective resistance estimation problem (see Section~\ref{sec:improved_lower_bounds_effres}) therefore also apply to the work of \cite{li2022new}.

\section{Preliminaries}
\label{section:preliminaries}\label{sec:preliminaries}

\paragraph{General notation.}{ We use $\A_{i,j}$ to denote the $(i,j)$-th entry of $\A$. For $\A \in \R^{n\times n}$; we use $\Spectrum{\A}$ for its spectrum; $\lambda_i(\A)$ and $\sigma_i(\A)$ for its $i$-th smallest eigenvalue and singular value respectively; and $\rho(\A) := \abs{\lambda_n(\A)}$ for its spectral radius. We use $\Vert \cdot\Vert_p$ for the $\ell_p$-norm. When $\A$ is PSD, $\lambdamin(\A)$ denotes its smallest nonzero eigenvalue and $\kappa(\A):=\lambda_n(\A)/\lambdamin(\A)$ denotes its pseudo-condition number. We use $\langle \cdot, \cdot \rangle$ for the Euclidean inner product; $\one$ for the all ones vector; and $e_i$ for the $i$-th standard basis vector. We define $\deltaij := e_i - e_j$ and $[k] \defeq \{1, ..., k\}$. We use $x \approx_\epsilon y$ as shorthand for $(1-\epsilon) y \leq x \leq (1+\epsilon) y$. For $v \in \R^n$, we use $v[i:j]$ for the sub-vector from index $i$ to $j$.}

\paragraph{Graphs.}{We use $G = (V, E, w)$ to denote a weighted undirected graph on $V$ with edges $E$ and edge weights $w \in \R^{E}_{> 0}$ (or $G = (V, E)$ if unweighted). We use $\A_G$ to denote its (weighted) adjacency matrix $(\A_G)_{u,v} = w_{u,v}$ for $u,v \in V \times V$ and $\D_G$ to denote its diagonal (weighted) degree matrix $(\D_G)_{u} = \sum_{\{u, v\} \in E} w_{u,v}$ for $u \in V$ (treated as $w_{u,v} = 1$ if $G$ is unweighted). We define $\laplacian{G} := \D_G - \A_G$ as its graph Laplacian. $\dmax(G)$ and ${\dmin}(G)$ refer to the max and min diagonal entry in $\D_G$. We may drop the argument or subscript $G$ it is clear from context. The effective resistance between nodes $i, j \in V$ is denoted $\rGij = \deltaij^\top \laplacian{G}^\dagger \deltaij$. We assume all input graphs are connected, as effective resistances can be computed separately on connected components.
}

\paragraph{Symmetric Diagonally Dominant (SDD) Matrices}{
A matrix $\M \in \R^{n \times n}$ is SDD if it can be decomposed as $\M = \diagonal{\M} - \adjacency{\M}$, where the $\diagonal{\M}$ is a diagonal matrix with non-negative entries and $\adjacency{M}$ is a matrix with zeros on the diagonal such that $d_{i,i} > \sum_{j = 1}^n |a_{i,j}|$. We define the normalization of $\M$ as $\normalized{M} := \diagonal{M}^{-1/2}\M \diagonal{M}^{-1/2}$. Throughout this paper, we assume, without loss of generality that $\D_M$ has strictly positive entries on the diagonal (otherwise, we can simply remove an entire row and column of zeros). We use ${\dmax}(\M)$ and ${\dmin}(\M)$ to denote the max and min entry in the diagonal of $\D_M$ respectively. We may drop the argument or subscript $\M$ if it is clear from context. }

\paragraph{Runtimes and Space Complexities.} {In our algorithmic results and analysis, when clear from context, we use $\widetilde{O}(\cdot)$ notation to hide polylogarithmic factors in the number of vertices, the number of edges, the size of the matrix, the number of nonzero entries in a matrix, the maximum and minimum diagonal element of a matrix, the maximum and minimum weighted degree of a graph, $\epsilon$, the condition number, and the normalized psuedo-condition number of a matrix (see Definition~\ref{def:pseudo-condition-number}). We say event $E$ occurs with high probability in $t$ if $\Prob{E} \geq 1 - t^{-c}$, where $c> 0$ can be controlled by appropriately configuring the algorithm parameters. We may simply say that an event occurs ``with high probability'' or ``whp.'' if it occurs with high probability in the dimension of a matrix or number of nodes in a graph.} 
\section{Algorithmic Results}\label{section:upper_bounds}

In this section, we present our main algorithmic results. Section~\ref{sec:ub_overview} outlines our approach to effective resistance sketches and estimation. Section~\ref{subsection:technical} gives an overview of a few prior technical results which we directly use in our algorithms. Section~\ref{subsection:upper_bound_new}, presents our original results on effective resistance sketching and estimation and generalizations to SDD matrices. Section~\ref{section:ub_extensions} extends our techniques to yield an interesting data structure for estimating quadratic forms of PSD matrices. 

\subsection{Our Approach}\label{sec:ub_overview}

\paragraph{Approach in prior work: Johnson Lindenstrauss sketches.} { Our starting inspiration is a classic result of \cite{spielman2009graph}, which obtains an ($\Otilde{m\epsilon^{-2}}, \Otilde{n\epsilon^{-2}}$, $\Otilde{\epsilon^{-2}}$) effective resistance sketch by using the Johnson Lindenstrauss Lemma (JL) \cite{JL} and its algorithmic instantiations \cite{achlioptas}. 

\begin{lemma}[Johnson-Lindenstrauss Lemma \cite{achlioptas}]\label{lemma:JL} Given fixed vectors $v_1, ..., v_n \in \R^d$ and $\epsilon \in (0, 1)$, let $\Q$ be an independently sampled random matrix in $\{\pm 1/\sqrt{k}\}^{k \times d}$. For $k = \Otilde{\log(n) \epsilon^{-2}}$, whp.\ $\Norm{\Q v_i}_2 \approx_\epsilon \Norm{v_i}_2$ for all $i \in [n]$. 
\end{lemma}

\cite{spielman2009graph} observe that $r_G(i, j) = (\W_G^{1/2}\B_G\laplacian{G}^\dagger)^\top (\W_G^{1/2}\B_G\laplacian{G}^\dagger)$, where $\W_G \in \R^{E\times E}$ is the diagonal matrix of weights in $G$, and $\B_G$ is the $E \times V$ edge-incidence matrix of $G$. Consequently, whp.\ $\Vert \J \W_G^{1/2}\B_G\laplacian{G}^\dagger\deltaij \Vert_2 \approx_\epsilon r_G(i, j)$. With SDD linear system solvers, $\J\W_G^{1/2}\B_G\laplacian{G}^\dagger$ can be approximated in $\Otilde{m\epsilon^{-1}}$ time, from which $\Vert \J\W_G^{1/2}\B_G\laplacian{G}^\dagger \deltaij \Vert_2$ can be queried in $\Otilde{\epsilon^{-2}}$ time. }

\paragraph{Our approach: asymmetric CountSketch in $\ell_1$.} {Towards improving upon JL sketches for effective resistance estimation, our key tool is to use that there are other sketching algorithms that achieve better than $\widetilde{O}(\epsilon^{-2})$ bit compressions of vectors vectors with small $\ell_1$ norm with comparable guarantees. CountSketch is a classic memory-efficient algorithm for estimating the number of occurences of various datapoints in a data stream \cite{CHARIKAR20043} and efficiently computing inner products \cite{LPT2012}. Given $v \in \R^n$ and integer parameters $s, t > 0$, CountSketch transforms $v$ to a vector $\Smat \widetilde{v} \in \R^{3ts \times n}$, where $\Smat \in \R^{3ts \times n}$ is a $3t$-column-sparse 0/1 matrix. Lemma~\ref{lemma:lemmaCountSketch} is a special case of Theorem 4 from \cite{LPT2012}, which provides accuracy guarantees for inner product estimation using CountSketch. }

\begin{restatable}[Special Case of \cite{LPT2012}, Theorem 4]{lemma}{lemmaCountSketch}\label{lemma:lemmaCountSketch} \label{lemma:8} Let vectors $v, w \in \mathbbm{R}^n$. Let $\Smat$ be a random CountSketch matrix. Let $x_i = \langle (\Smat v)[(i-1)s+1:i\cdot s], (\Smat w)[(i-1)s+1:i\cdot s]\rangle$ for $i \in [3t]$, and let $X$ denote the median of $\{x_i\}$. For $s = O\left( \min\left( \frac{\Vert v \Vert_1 \Vert w \Vert_1}{\epsilon \abs{\langle v, w \rangle}}, \frac{\Vert v \Vert_2^2 \Vert w \Vert_2^2}{\epsilon^2 \abs{\langle v, w \rangle^2}} \right)\right)$, and $t = \log(n^{c})$, $\abs{X - \langle v, w \rangle} \leq \epsilon \abs{\langle v, w \rangle}$ with probability at least $\Omega(1-n^{-c})$.
\end{restatable}
\begin{proof}[Proof of Lemma~\ref{lemma:lemmaCountSketch}] Applying Theorem 4 from \cite{LPT2012} with $t = 2$,
\begin{align*}
    \mathbbm{E}\left[\left\lvert X-\left\langle v, w \right\rangle \right\rvert\right] = 0, \quad 
    \mathbbm{V}\left[\left\lvert X-\left\langle v, w \right\rangle \right\rvert\right] \leq \min\left(3 \frac{\Vert v \Vert_1^2 \Vert w \Vert_1^2}{s^2}, 2\frac{\Vert v \Vert_2^2 \Vert w \Vert_2^2}{s}\right)
\end{align*}
Applying Chebyschev's inequality, 
\begin{align*}
    \mathbbm{P}\left[\left\lvert X-\langle v, w \rangle \right\rvert > 2\sqrt{3} \min\left(\frac{\Vert v \Vert_1 \Vert w \Vert_1}{s}, \frac{\Vert v \Vert_2 \Vert w \Vert_2}{\sqrt{s}}\right)\right] \leq \frac{1}{4}. 
\end{align*}
Consequetly, setting $s = O\left( \min\left(\frac{\Vert v \Vert_1 \Vert w \Vert_1}{\epsilon \langle v, w \rangle}, \frac{\Vert v \Vert_2^2 \Vert w \Vert_2^2}{\epsilon^2 \langle v, w \rangle^2}\right)\right)$ suffices for \\ $\mathbbm{P}\left[\abs{X -\langle v, w\rangle} \leq \epsilon \langle v, w \rangle \right] \geq 3/4$. To improve the failure probability to $O(n^{-c})$, it suffices to use the median trick, i.e., repeat the sketch $O\pars{\log(n^c)}$ times and output the median. 
\end{proof}

To improve the guarantee in Lemma~\ref{lemma:lemmaCountSketch} to hold whp.\ \emph{for all} $v, w \in S$ rather than for each fixed pair, one can simply choose $t = \log(n^{c}|S|)$ and apply a union bound. Consequently, if we knew that $\Vert \J\W_G^{1/2}\B_G\laplacian{G}^\dagger \deltaij\Vert_1^2/ r_G(i,j) = \Otilde{1}$, then building a CountSketch with $s = \Otilde{1}$ would yield a $\Otilde{n\epsilon^{-1}}$-size sketch, improving over the $\Otilde{n\epsilon^{-2}}$ sketch obtained using the $\ell_2$ JL sketch in \cite{spielman2009graph}. Unfortunately, it is unclear if and when such a bound holds, and so, it is unclear how the $\ell_1$ CountSketches could be useful in this setting. This leads to the main insight that fuels our algorithms. Rather than seeking a symmetric factorization of $r_G(i, j)$ as a quadratic form $v^\top v$ and applying a sketching procedure to $v$, we instead work with the following \emph{asymmetric factorization}:
\begin{align}\label{eq:asym}
    r_G(i,j) = \frac{1}{2}\langle \diagonal{\laplacian{G}}^{-1} \deltaij, \diagonal{\mlap}^{1/2} \left(\frac{\normalized{\mlap}} {2}\right)^\dagger \diagonal{\laplacian{G}}^{-1/2} \deltaij \rangle.
\end{align}

At first glance, it may seem unclear why \eqref{eq:asym} is helpful. However, we show that indeed, for expanders 
\begin{align}\label{eq:res}
    \normFull{ \diagonal{\mlap}^{-1} \deltaij}_1 \normFull{\diagonal{\mlap}^{1/2} \left(\frac{\normalized{\mlap}} {2}\right)^\dagger \diagonal{\mlap}^{-1/2} \deltaij }_1 /r_G(i,j) = \Otilde{1}. 
\end{align}
Our main result essentially follows from \eqref{eq:res}. Using SDD linear system solvers, we can efficiently approximate $\Smat \diagonal{\mlap}^{1/2} \left(\frac{\normalized{\mlap}} {2}\right)^\dagger \diagonal{\mlap}^{-1/2} \deltaij  \in \R^{\Otilde{\epsilon^{-1}} \times n}$ in $\Otilde{m\epsilon^{-1}}$ time, yielding our $\timeQuery$ of $\Otilde{m\epsilon^{-1}}$ and $s$ of $\Otilde{n\epsilon^{-1}}$. Moreover, $\Smat \diagonal{\mlap}^{-1}$ is $\Otilde{1}$-sparse. Consequently, using our (approximate) access to $\Smat \diagonal{\mlap}^{1/2} \left(\frac{\normalized{\mlap}} {2}\right)^\dagger \diagonal{\mlap}^{-1/2}$, for any query $i, j \in V$, we can efficiently approximate \eqref{eq:asym} using the recovery procedure of Lemma~\ref{lemma:lemmaCountSketch} in $\Otilde{1}$ time. 

\subsection{Technical Prerequisites}\label{subsection:technical}

In this section we will formally restate two fundamental results from prior literature which enable our upper bounds.

\paragraph{SDD Linear System Solvers.}{
In order to compute our effective resistance sketches efficiently, we apply a CountSketch matrix $\Smat$ to $\M^\dagger$, where $\M$ is an SDD matrix. To do this efficiently, we leverage the following theorem.

\begin{theorem}[SDD Linear System Solver]\label{theorem:solver}  Let $\M \in \R^{n \times n}$ be SDD and consider any $\beta > 0$. There exists a randomized algorithm which, whp. processes a graph in time $\widetilde{O}(\nnzFull{\M})$ to create access to a linear operator $\Q_{\beta} \in \mathbbm{R}^{n \times n}$ such that $\Q_{\beta}$ can be applied to any $b \in \mathbbm{R}^n$ with $b \perp \nullspace{\M}$ in time $\widetilde{O}(m)$ and $\left\Vert \Q_\beta b - \M^\dagger b \right\Vert_\M \leq \beta \left\Vert \M^\dagger b \right\Vert_\M$. 
\end{theorem}

Many SDD linear system solvers can be viewed as the type of an operator $\Q_\beta$ required in Theorem~\ref{theorem:solver}. For a particular example in which this is apparent, consider the operator corresponding to the iterative solver proposed in \cite{RS16} or the solver from \cite{KOSZ13}. There is a long line of research on nearly linear time SDD and Laplacian system solvers, beginning with the work of \cite{spielman2008local} and leading to current state-of-the-art randomized algorithm of \cite{jambulapati2023ultrasparse}. 
}

\paragraph{Cheeger's Inequality}{
To connect our results on well-conditioned SDD matrices to expander graphs, we use the well-known Cheeger's inequality \cite{Cheeger}. 
\begin{theorem}[Cheeger's Inequality]\label{thm:cheeger} Let $G = (V, E, w)$ be an undirected graph. Then, $\frac{1}{2} \lambda_2(\normalized{L}) \leq \phi(G) \leq \sqrt{2 \lambda_2(\normalized{L})}$. 
\end{theorem}}

\subsection{Effective Resistance Estimation in Expanders and Extensions to SDD matrices }\label{subsection:upper_bound_new}

We use our approach from Section~\ref{sec:ub_overview} to develop algorithms to compute spectral sketches for pseudoinverses of SDD matrices $\M$ with $\Otilde{1}$ normalized condition number, as defined in Definitions \ref{def:pseudo-condition-number} and \ref{def:spectral_sketch}.

\begin{definition}[Normalized (pseudo-)condition number]\label{def:pseudo-condition-number} We define the \emph{normalized (pseudo-)condition number} of an SDD matrix $\M \in \R^n$ as
$
    \pcondOf{\M} \defeq \lambda_n(\normalized{M})/\lambdamin(\normalized{M})
$.
\end{definition}

By Cheeger's inequality \cite{Cheeger}, if $G$ has $\alpha = \widetilde{\Omega}(1)$ expansion, then 
\begin{align*}
    \pcondOf{\laplacian{G}} = \lambda_n(\normalized{L_G})/\lambda_2(\normalized{L_G}) \leq4/\alpha^2 = \Otilde{1}
\end{align*}
So our spectral sketch algorithms will apply to Laplacians of expanders. 

\begin{definition}[Spectral Sketch Data Structure]\label{def:spectral_sketch} We say an algorithm produces a $(\timeSketch, \timeQuery, s)$-spectral sketch for a PSD matrix $\A \in \R^{n \times n}$ if given $\A \in \R^{n \times n}$ and $\epsilon \in (0, 1)$, the algorithm creates a binary string of length $O(s(A, \epsilon))$ in time $O(\timeSketch(A, \epsilon)$, from which, for a query $b \in \R^n$, whp.\ it outputs $q_A(b) \approx_\epsilon b^\top \A b$ in time $O(\timeQuery(A, \epsilon) \nnz(b))$. 
\end{definition}

Our spectral sketches of SDD matrices $\M$ will \emph{only} support $\diagonal{M}$-numerically sparse query vectors, defined as follows.
\begin{definition}[$\D$-numerically sparse]\label{def:D_numerically_sparse} For a diagonal matrix $\D$, the $\D$-numerical sparsity of $x \in \R^n$ is $\nsD{\D}{x} := \Norm{ \D^{-1} x}_1 \Norm{x}_1/\Norm{\D^{-1/2} x}_2^2$. We say $x$ is $(c, \D)$-numerically sparse if $\nsD{\D}{x} \leq c$. 
\end{definition}

Definition~\ref{def:D_numerically_sparse} may appear restrictive; however, several natural classes of vectors satisfy the requirements, including the following examples: 
\begin{itemize}[leftmargin=*]
    \item $\deltaij$ is (2, $\D$)-numerically sparse for any invertible $\D \in \R^{n \times n}$. As we are interested primarily in effective resistance estimation in this paper, this provides the primary motivation for studying this class of vectors. 
    \item Standard basis vectors are (1, $\D$)-numerically sparse for any invertible $\D \in \R^{n \times n}$. Note that this implies our algorithms can be used to efficiently compute the diagonals of pseudoinverses of well-conditioned SDD matrices. 
    \item When $\D$ is the identity matrix, Definition~\ref{def:D_numerically_sparse} reduces to the standard definition of numerical sparsity \cite{GS2018}. 
    \item More generally, if $x \in \R^n$ is $\gamma$-numerically sparse, then it is $\left(\gamma {\frac{\max_{i} d_{i,i}}{\min_{i} d_{i,i}}}, \D\right)$-numerically sparse for any diagonal $\D \in \R^{n \times n}$. Note that if $\D$ is approximately a multiple of the identity, then the $\diagonal{\M}$-numerical sparsity is approximately equal to the numerical sparsity, up to constants. 
\end{itemize}

The following asymmetric rearrangement of quadratic forms is crucial to our analysis. 

\begin{restatable}{lemma}{lemmaEquivalentQForms}\label{lemma:lemmaEquivalentQForms} Let $\M \in \R^{n \times n}$ be SDD and $x \in \R^n$ be peridicular to $\nullspace{\M}$. Then, 
\begin{align*}
    x^\top \M^\dagger x &= \frac{1}{2}\left\langle \diagonal{M}^{-1}x, \diagonal{M}^{1/2} (\normalized{M}/2)^{\dagger} \diagonal{M}^{-1/2} x\right\rangle = \left\langle \diagonal{M}^{-1/2}x, \normalized{M}^{\dagger} \diagonal{M}^{-1/2} x\right\rangle 
    \geq \frac{1}{2} \normFull{\diagonal{M}^{-1}x}. 
\end{align*}
\end{restatable}
\begin{proof} For notational convenience, let $\normalizedlazy{M} = \normalized{M}/2$. Let $v = \M^\dagger x = (\diagonal{M}-\adjacency{M})^\dagger x$. Note that $\diagonal{M}^{-1/2}x \perp \nullspace{\normalized{M}}$, and $2 \diagonal{M}^{1/2} \normalizedlazy{M} \diagonal{M}^{1/2} = \M$. Consequently, $v = \frac{1}{2} \diagonal{M}^{-1/2}\normalizedlazy{M}^\dagger \diagonal{M}^{-1/2} x$, and hence 
\begin{align*}
    x^\top \M^\dagger x = \frac{1}{2} \langle {\diagonal{M}}^{-1}x, {\diagonal{M}}^{1/2} \normalizedlazy{M}^{\dagger} {\diagonal{M}}^{-1/2} x\rangle.
\end{align*}
The second equality now follows immediately by rearranging terms. To obtain the inequality, note that, because $\M$ is SDD and $\D$ is invertible, $\normalized{M}$ is PSD. Furthermore, $2\I - \normalized{M} = \I - \diagonal{M}^{-1/2} \adjacency{M} \diagonal{M}^{-1/2}$, which is also PSD, as $\Spectrum{\adjacency{M}} \subset [-\dmax, \dmax]$. So, $\Spectrum{\normalized{M}} \leq 2$. So, $\lambdamin(\normalized{M}) \geq 1/2$ and the lemma follows.  
\end{proof}

Now, we wish to bound $\Norm{\diagonal{M}^{1/2} (\normalized{M}/2)^{\dagger} \diagonal{M}^{-1/2} x}_1$. To do this, we leverage the power series expansion of $(\normalized{M}/2)^{\dagger}$ as follows. In the following Theorem, for notational convenience we let $\adjacencylazy{M}:= \I - (\normalized{M}/2)$. 

\begin{lemma}\label{lemma:powerseries} Let $\M \in \R^{n \times n}$ be SDD. For any $x \perp \nullspace{M}$, 
\begin{align*}
    \left(\normalized{M}/2\right)^\dagger \diagonal{M}^{-1/2} x &= \sum_{k=0}^\infty \left(\adjacencylazy{M} \right)^k \diagonal{M}^{-1/2} x,
\end{align*}
and for $m \geq 1$, 
\begin{align*}
        \left\Vert \left(\normalized{M}/2\right)^\dagger \diagonal{M}^{-1/2}x - \sum_{k=0}^{m} \left( \adjacencylazy{M} \right)^k  \diagonal{M}^{-1/2}x \right\Vert_1 &\leq \frac{\sqrt{n} 2\pcondOf{M}}{\sqrt{\dmin}} \exp\left(- \frac{m+1}{2\pcondOf{M}}\right) \norm{x}_2.
\end{align*}

\end{lemma}
\begin{proof} Let $r$ denote the rank of $M$. Let $q_1, ..., q_{n}$ denote orthonormal eigenvectors of $\adjacencylazy{M}$ associated with $\lambda_1(\adjacencylazy{M}), ..., \lambda_{n}(\adjacencylazy{M})$ respectively, $\Q$ denote the orthogonal matrix whose $i$-th column is $q_i$, and $\widetilde{\Lambda}$ denote the diagonal matrix of the $\lambda_i(\adjacencylazy{M})$'s. 

An orthogonal eigendecomposition of ${(\normalized{M}}/2)^\dagger$ is given by  $(\normalized{M}/2)^\dagger = \Q \Lambda \Q^\top$, where $\Lambda$ is the diagonal matrix whose entries are given by 
\begin{align*}
    \Lambda_{i,i} = \begin{cases}
        0, & i = r+1, ..., n \\
        \frac{1}{1-\lambda_i(\adjacencylazy{M})}, & i = 1, ..., r
    \end{cases}. 
\end{align*}
So, for any $t \in [r]$, ${(\normalized{M}/2})^\dagger q_t = \frac{1}{1-\lambda_t(\adjacencylazy{M})} q_t$, where $\lambda_t(\adjacencylazy{M}) \in (0, 1)$; and consequently, $0 \leq \lambda_t(\adjacencylazy{M}) < 1$ is in the radius of convergence for the power series of $\frac{1}{1-x}$, and hence 
\begin{align*}
    \sum_{k=0}^\infty \left(\adjacencylazy{M}\right)^k = \sum_{k=0}^\infty \Q \widetilde{\Lambda}^k \Q^\top q_t = q_t \sum_{k=0}^\infty \lambda_{t}(\adjacencylazy{M})^k = \frac{1}{1-\lambda_t(\adjacencylazy{M})} q_t. 
\end{align*}
Now, $x \perp \nullspace{M}$ implies $\diagonal{M}^{-1/2}x \perp \nullspace{\normalized{M}}$; and consequently, $\diagonal{M}^{-1/2}x$ can be expressed as a linear combination of $q_1, ..., q_r$. The first statement in the lemma now follows by linearity. 

For the second statement, note that $\lambdamin(\normalizedlazy{M}) = \lambda_n(\normalizedlazy{M})/\pcondOf{M} \geq {1}/{(2\pcondOf{M})}$, so $\lambda_{r}\left( \adjacencylazy{M} \right) \leq 1 - {1}/{(2 \pcondOf{M})}$. Since $\diagonal{M}^{-1/2}x \perp q_{r+1}, ..., q_{n}$,
\begin{align*}
    \left\Vert (\adjacencylazy{M})^k \diagonal{M}^{-1/2}x \right\Vert_2 \leq \left(1 - \frac{1}{2\pcondOf{M}} \right)^k \frac{\normFull{x}_2}{\sqrt{\dmin}}. 
\end{align*}
Using the fact that $\norm{x}_1 \leq \sqrt{n}\norm{x}_2$ for all $x$, for $m \geq 1$, we have
\begin{align*}
    \left\Vert \left(\normalized{M}/2\right)^\dagger \diagonal{M}^{-1/2}x - \sum_{k=0}^{m} \left( \adjacencylazy{M} \right)^k  \diagonal{M}^{-1/2}x \right\Vert_1
    &\leq \sqrt{n} \sum_{k=m+1}^{\infty} \left\Vert \left( \adjacencylazy{M} \right)^k  \diagonal{M}^{-1/2}x \right\Vert_2 \\
    &\leq \frac{\sqrt{n} 2\pcondOf{M}}{\sqrt{\dmin}} \left(1 - \frac{1}{2\pcondOf{M}} \right)^{m+1} \norm{x}_2 \\
    &\leq \frac{\sqrt{n} 2\pcondOf{M}}{\sqrt{\dmin}} \exp\left(- \frac{m+1}{2\pcondOf{M}}\right) \norm{x}_2.
\end{align*}
\end{proof}

Using Lemma~\ref{lemma:powerseries}, the following bound follows almost immediately.

\begin{restatable}{lemma}{lemmaloneNormBound}\label{lemma:lemmal1NormBound} Let $\M \in \R^{n \times n}$ be SDD matrix and $x \in \R^n$ be a unit vector orthogonal to $\nullspace{\M}$. Let $m = \max\left(1, 2\pcondOf{\M}\log\left({\sqrt{n\dmax} 2\pcondOf{\M}}/{\sqrt{\dmin}}\right)\right)$. Then 
\begin{align*}
    \normFull{\diagonal{M}^{1/2} (\normalized{M}/2)^{\dagger}  {\diagonal{M}}^{-1/2} x}_1 \leq m \norm{x}_1 + 1.
\end{align*}
\end{restatable}

\begin{proof} Let $m \geq \max\left(1, 2\pcondOf{M}\log\left(\frac{\sqrt{n\dmax} 2\pcondOf{M}}{\sqrt{\dmin}}\right)\right)$. By Lemma~\ref{lemma:powerseries}, 
\begin{align*}
    \left\Vert \diagonal{M}^{1/2} \normalizedlazy{M}^\dagger \diagonal{M}^{-1/2}x - \sum_{k=0}^{m} D^{1/2 }{\adjacencylazy{M}}^k \diagonal{M}^{-1/2} x \right\Vert_1 
    &= \left\Vert \diagonal{M}^{1/2}\left(\normalizedlazy{M}^\dagger \diagonal{M}^{-1/2}x - \sum_{k=0}^{m}  {\adjacencylazy{M}} ^k  \diagonal{M}^{-1/2}x \right) \right\Vert_1 \\
    &\leq \frac{\sqrt{n\dmax} 2\pcondOf{M}}{\sqrt{\dmin}} \exp\left(- \frac{m+1}{2\pcondOf{M}}\right) \norm{x}_2 
    \leq 1. 
\end{align*}
Using triangle inequality plus the observation that $\diagonal{M}^{1/2}{\adjacencylazy{M}} \diagonal{M}^{-1/2} =  1/2 \I + 1/2{\adjacency{M}} \diagonal{M}^{-1}$ has each column normalized to have absolute column sum at most 1,
\begin{align*}
    \left\Vert \diagonal{M}^{1/2} \normalizedlazy{M}^\dagger \diagonal{M}^{-1/2}x \right\Vert_1 \leq \left\Vert \sum_{k=0}^{m} \diagonal{M}^{1/2}{\adjacencylazy{M}}^k \diagonal{M}^{-1/2} x \right\Vert_1 + 1 \leq m \Vert x \Vert_1 + 1.
\end{align*}
\end{proof}

Combining our Lemmas~\ref{lemma:lemmaEquivalentQForms} and \ref{lemma:lemmal1NormBound} with the guarantees of Lemma~\ref{lemma:lemmaCountSketch} and prior work on SDD linear system solvers, we obtain the following theorem.

\begin{restatable}{theorem}{theoremSDDMain}\label{theorem:theoremSDDMain} For any matrix $\M \in \R^{n \times n}$ with $\pcondOf{\M} = \Otilde{1}$, there is an algorithm which is an $(\Otilde{\nnz(\M)\epsilon^{-1}}, \Otilde{1}, \Otilde{n\epsilon^{-1}})$ spectral of $\M^\dagger$ supported over queries $S$, where $S$ is any set of $(\Otilde{1}, \diagonal{M})$-numerically sparse vectors orthogonal to $\nullspace{\M}$. 
\end{restatable}

\begin{proof} It suffices to assume, without loss of generality, that $S$ is a set of unit vectors, as at query time, for any vector $b$ we can compute $\norm{b}_2$ in $O(\nnz(b))$ time and rescale. Set $\beta = \frac{ 2\min(1, \dmin^{3}) \lambdamin(\M) \epsilon}{\max(1, \dmax^2) \sqrt{n\max(1, \dmax)}}$. By Theorem~\ref{theorem:solver}, in $\Otilde{\nnzFull{\M}}$, whp.\ we can obtain access to a linear operator $\Q_\beta$ such that $\Q_\beta$ can be applied to any $b \in S$ in time $\Otilde{\nnz(\M)}$ and $\norm{\Q_\beta b - \M^\dagger b}_M \leq \beta \norm{\M^\dagger b}_M = \beta \norm{b}_{\M^\dagger}$. Then, 
\begin{align*}
    \lambdamin(\M) \norm{(\Q_\beta b - \M^\dagger)b}_2^2 \leq \norm{(\Q_\beta b - \M^\dagger)b}_M^2 \leq \beta^2 \norm{b}_{M^\dagger}^2 \leq \frac{\beta^2}{\lambdamin(\M)} \norm{b}_2^2. 
\end{align*}
So, $ \norm{(\Q_\beta b - \M^\dagger)b}_2 \leq \frac{\beta}{\lambdamin(\M)} \leq \epsilon$. Consequently, by triangle inequality, we have that 
\begin{align*}
    \norm{2\diagonal{M} \Q_\beta b}_1  &\leq \norm{2\diagonal{M}\M^\dagger b}_1 + \norm{2\diagonal{M}\Q_\beta b -  2\diagonal{M}\M^\dagger b}_1
\end{align*}
where $2\diagonal{M}\M^\dagger b = \diagonal{M}^{1/2}(\normalized{M}/2)^\dagger \diagonal{M}^{1/2} b$ and $\norm{2\diagonal{M}\Q_\beta b -  2\diagonal{M}\M^\dagger b}_1 \leq 2\sqrt{n \dmax} \norm{\Q_\beta b -  \M^\dagger b}_2 \leq \epsilon$. Consequently, 
\begin{align*}
    \norm{2\diagonal{M} \Q_\beta b}_1  &\leq \norm{\diagonal{M}^{1/2}(\normalized{M}/2)^\dagger \diagonal{M}^{1/2} b}_1 + \epsilon.  
\end{align*}
Now, Lemma~\ref{lemma:lemmal1NormBound} guarantees that $\norm{2\diagonal{M} \Q_\beta b}_1 \leq \Otilde{\pcondOf{\M} + \epsilon} \norm{b}_1 = \Otilde{\pcondOf{\M}}\norm{b}_1$. Similarly, 
\begin{align*}
    \abs{\langle \diagonal{M}^{-1} b, \diagonal{M}^{1/2}(\normalized{M}/2)^\dagger \diagonal{M}^{-1/2} b \rangle - \langle \diagonal{M}^{-1}b, 2\diagonal{M} \Q_\beta b\rangle } &= \abs{\langle \diagonal{M}^{-1} b, 2\diagonal{M}\M^\dagger b \rangle - \langle \diagonal{M}^{-1}b, 2\diagonal{M} \Q_\beta b\rangle } \\
    &\leq \dmax \normFull{\diagonal{M}^{-1} b}_2  \norm{(\Q_\beta b - \M^\dagger)b}_2 \\
    &\leq \frac{\dmax}{\dmin} \frac{\beta}{\lambdamin(\M)} 
    \leq \epsilon \left( \frac{1}{2\dmin^2} \right).
\end{align*}
Lemma~\ref{lemma:lemmaEquivalentQForms} guarantees that $\langle \diagonal{M}^{-1}b, 2\diagonal{M} \Q_\beta b\rangle \geq O(1) \norm{\diagonal{\M}^{-1}b}_2^2$. Consequently, 
\begin{align*}
    \frac{\norm{\diagonal{M}^{-1}b}_1 \norm{2\diagonal{M} \Q_\beta b}_1}{\langle \diagonal{M}^{-1/2}b, 2\diagonal{M} \Q_\beta b\rangle} = \Otilde{\pcondOf{\M}} \frac{\norm{\diagonal{M}^{-1}b}_1 \norm{b}_1}{\norm{\diagonal{\M}^{-1}b}_2^2} =\Otilde{\pcondOf{\M}} \nsD{\diagonal{M}}{b} = \Otilde{\pcondOf{\M}}.
\end{align*}

So, given a CountSketch matrix $\Smat \in \R^{\Otilde{\pcondOf{\M}\epsilon^{-1}} \times n}$, Lemma~\ref{lemma:lemmaCountSketch} guarantees that we can compute an $X$ such that whp.\, 
\begin{align*}
    \abs{X - \langle \diagonal{M}^{-1} b, 2\diagonal{M} \Q_\beta b\rangle} \leq \epsilon \langle \diagonal{M}^{-1} b, 2\diagonal{M} \Q_\beta b\rangle.
\end{align*}
Moreover, we showed above that
\begin{align*}
     \abs{\langle \diagonal{M}^{-1} b, 2\diagonal{M} \Q_\beta b\rangle - \langle \diagonal{M}^{-1} b, \diagonal{M}^{1/2}\left(\normalized{M}/2\right)^\dagger \diagonal{M}^{-1/2} b \rangle} &\leq \epsilon \left( \frac{1}{2 \dmin ^2} \right) \\
     &\leq \epsilon \langle \diagonal{M}^{-1} b, \diagonal{M}^{1/2}\left(\normalized{M}/2\right)^\dagger \diagonal{M}^{-1/2} b \rangle,  
\end{align*}

where the last line uses the observation from Lemma~\ref{lemma:lemmaEquivalentQForms}, that $\langle \diagonal{M}^{-1} b, \diagonal{M}^{1/2}(\normalized{M}/2)^\dagger \diagonal{M}^{-1/2} b \rangle \geq \frac{1}{2} \norm{\diagonal{\M}^{-1}b}_2^2$. It now follows that
\begin{align*}
    \abs{X - \langle \diagonal{M}^{-1} b, \diagonal{M}^{1/2}(\normalized{M}/2)^\dagger \diagonal{M}^{-1/2} b \rangle } &\leq \epsilon \langle \diagonal{M}^{-1} b, 2\diagonal{M} \Q_\beta b\rangle + \epsilon \langle \diagonal{M}^{-1} b, \diagonal{M}^{1/2}\left(\normalized{M}/2\right)^\dagger \diagonal{M}^{-1/2} b \rangle \\
    &\leq 2\epsilon(1+\epsilon) \langle \diagonal{M}^{-1} b, \diagonal{M}^{1/2}\left(\normalized{M}/2\right)^\dagger \diagonal{M}^{-1/2} b \rangle \\
    &\leq 4\epsilon \langle \diagonal{M}^{-1} b, \diagonal{M}^{1/2}\left(\normalized{M}/2\right)^\dagger \diagonal{M}^{-1/2} b \rangle.
\end{align*}

Thus, by Lemma~\ref{lemma:lemmaEquivalentQForms}, $\frac{1}{2} X \approx_{4\epsilon} \langle \frac{1}{2} \diagonal{M}^{-1} b, \diagonal{M}^{1/2}\left(\normalized{M}/2\right)^\dagger \diagonal{M}^{-1/2} b \rangle = \langle b, M^\dagger b \rangle$; hence, rescaling $\epsilon$ by a constant factor of 4 yields the approximation guarantee (without changing the size of $\Smat$ by more than constant factors). 

Consequently, we see that by storing $\Smat$ and $\Q_\beta$, we can support queries $b \in S$. To justify the runtime guarantee, note that Theorem~\ref{theorem:solver} shows we can compute $\Smat \Q_\beta$ in $\Otilde{\nnz(\M)\pcondOf{\M}\epsilon^{-1}}$ time, by applying an approximate SDD system solver to each row in $\Smat$. To support queries, we need only store $\Smat \Q_\beta$ and $\Smat$, which requires only $\Otilde{n \pcondOf{\M} \epsilon^{-1}}$ bits. 

Finally, we justify the query complexity. The key observation is that $\Smat$ is $\Otilde{1}$-column sparse. Computing $X$ requires taking the median of $\Otilde{1}$ quantities, each of which requires computing an inner product involving an $\Otilde{\nnz(b)}$-sparse vector $\Smat \diagonal{M}^{-1} b$. Using this fact, $X$ can be computed in $\Otilde{\nnz(b)} $. Setting $\pcondOf{\M} = \Otilde{1}$ completes the proof.
\end{proof}

\begin{remark} Note that the proof of Theorem~\ref{theorem:theoremSDDMain} shows that, more broadly, for any SDD matrix $\M$, we can obtain an $(\Otilde{\pcondOf{\M} \nnz(\M) \epsilon^{-1}}, \Otilde{1}, \Otilde{\pcondOf{\M} n \epsilon^{-1}})$ spectral sketch data structure. 
\end{remark}

The corresponding algorithm pseudocode for constructing the spectral sketch data structure is given in Algorithm~\ref{alg:sketchSDD}. The algorithm pseudocode for querying the spectral sketch data structure is given in Algorithm~\ref{alg:querySDD}. The proof of Theorem~\ref{theorem:theoremSDDMain} guarantees that it suffices to set $t = \Otilde{1}$ and $s = \Otilde{\pcondOf{\M}\epsilon^{-1}}$ in Algorithm~\ref{alg:sketchSDD} and Algorithm~\ref{alg:querySDD}.

\begin{algorithm}[ht]
    \SetAlgoLined
    \KwIn{Matrix $\M \in \R^{n \times n}$, error tolerance $\epsilon \in (0, 1)$, and integers $s, t > 0$.}
    \KwOut{CountSketch matrix $\Smat \in \R^{(3t)s \times n}$ and an approximation of $2\Smat \diagonal{\M} \M^\dagger$, $\widetilde{\Smat} \in \R^{(3t)s \times n}$}
    Set $\beta$ as in the proof of Theorem~\ref{theorem:theoremSDDMain}\; 
    Generate a random CountSketch matrix $\Smat \in \R^{(3t)s \times n}$ (see Theorem 4 of \cite{LPT2012})  
    \; 
    Generate access to a linear operator $\Q_\beta$ such that for all $b \perp \nullspace{\M}$,  $\Vert \Q_\beta \b - \M^\dagger \b \Vert_M \leq \beta \norm{M^\dagger b}_\M$ (See Theorem~\ref{theorem:solver})\; 
    Compute $\widetilde{\Smat} = 2\Smat \diagonal{\M}\Q_\beta$\; 
    \Return{$(\Smat, \widetilde{\Smat})$}
    \caption{SpectralSketchSDDInverse}
    \label{alg:sketchSDD}
\end{algorithm}

\begin{algorithm}[ht]
    \SetAlgoLined
    \KwIn{$\diagonal{\M}$, output of Algorithm~\ref{alg:sketchSDD}, query vector $b$, and integers $s, t > 0$ as inputted to Algorithm~\ref{alg:sketchSDD}. }
    \KwOut{Approximation to $\langle b, \M^{\dagger} b\rangle$. }
    Set $\widetilde{b} = b/\norm{b}_2$\;
    \For{$i \in [3t]$}{
        $x_i = \frac{1}{2} \langle (\Smat \diagonal{M}^{-1}\widetilde{b})[(i-1)s+1 : (i)s]i, (\widetilde{\Smat} \widetilde{b})[(i-1)s+1 : (i)s] \rangle$ \; 
    }
    \Return{$\norm{b}_2^2 \mathrm{median}\{x_i\}$.}
    \caption{QuerySketchSDDInverse}
    \label{alg:querySDD}
\end{algorithm}

Because the $\deltaij$ queries appearing in effective resistance computations are 2-sparse and $(2, \diagonal{M})$ numerically sparse for all SDD matrices $\M$, taking $\M = \laplacian{G}$ in Theorem~\ref{theorem:theoremSDDMain} immediately implies Theorem~\ref{thm:ub_main_informal} and Theorem~\ref{thm:ub_main_informal2}.

\begin{proof}[Proof of Theorem~\ref{thm:ub_main_informal} and Theorem~\ref{thm:ub_main_informal2} ] Let $G = (V, E, w)$ be a graph with $\widetilde{\Omega}(1)$-expansion. As argued previously (see Section~\ref{subsection:upper_bound_new}, $\laplacian{G}$ is SDD with $\Otilde{1}$ normalized condition number, $\deltaij \perp \ker{\laplacian{G}}$, and $\deltaij$ is $(2, \diagonal{\laplacian{G}})$- numerically sparse. To see why Theorem~\ref{thm:ub_main_informal} holds, simply observe that we can boost the whp.\ guarante guarantee in Theorem~\ref{theorem:theoremSDDMain}, which holds for each fixed query, to hold whp.\ for all $\deltaij$ for $i, j \in V$ by maintaining $O(\log(n))$ copies of the spectral sketch data structure guaranteed in Theorem~\ref{theorem:theoremSDDMain} as our sketch, and then, at query time, taking the median of the results of the outputs from querying each of the $O(\log(n))$ copies of the sketch. Theorem~\ref{thm:ub_main_informal2} is a direct corollary of Theorem~\ref{thm:ub_main_informal}. 
\end{proof}

\subsection{Extensions to PSD Matrices}\label{section:ub_extensions}

Our approach of approximating quadratic forms via asymmetric inner products also yields a query-efficient sketching procedure for approximating quadratic forms of well-conditioned PSD matrices.

\begin{restatable}{theorem}{theoremPSDMain}\label{theorem:theoremPSDMain} There is an algorithm which is an $(\Otilde{\kappa(\A) \nnz(\A)\epsilon^{-2}}$, $\Otilde{1}, \Otilde{\kappa(\A)n\epsilon^{-2}})$ spectral sketch data of $\A$ supported on $\X$, where $\A$ is PSD and $\X$ is orthogonal to $\nullspace{\A}$.  
\end{restatable}

\begin{proof}[Proof of Theorem~\ref{theorem:theoremPSDMain}] 

Let $x \perp \nullspace{\A}$. Note that for any $\Norm{x}_2 \Norm{\A x}_2 = \kappa(\A^{1/2})x^\top \A x$. So, by the $\ell_2$ norm guarantees from Lemma~\ref{lemma:lemmaCountSketch}, it suffices to build a CountSketch matrix $S$ with $s = \widetilde{O}\left(\kappa({\A})\epsilon^{-2}\right)$ and $t = \Otilde{1}$ to guarantee that for any $x \perp \nullspace{\A}$, $\langle \Smat x, \Smat \A x \rangle \approx_\epsilon \langle x, \A x \rangle$ whp.\ The time to compute the sketch $\Smat \A$ is at most $\Otilde{\kappa(\A)\nnz(\A)\epsilon^{-2}}$. Due to the $\Otilde{1}$ column-sparsity of $\Smat$, $\Smat x$ is $\Otilde{\nnz(x)}$ sparse, and consequently, $\langle \Smat x, \Smat \A x \rangle$ can be computed in $\Otilde{\nnz(x)}$ time. 
\end{proof}

In comparison, JL gives an $(\Otilde{n^\omega}, \Otilde{n\epsilon^{-2}}, \Otilde{\epsilon^{-2}})$-spectral sketch data structure using efficient square root algorithms \cite{jain2017global}. JL achieves better compression than Theorem~\ref{theorem:theoremPSDMain}. For well-conditioned matrices, however, our query time and potentially construction times may be faster in some regimes.

\section{Lower Bounds}
\label{sec:lower_bounds}

In this section, we present our main conditional hardness results. In Section~\ref{sec:lb_overview}, we outline our approach. In Section~\ref{sec:improved_lower_bounds_effres}, we present our lower bounds for the problem of estimating effective resistances for all pairs of nodes (case where $S = V \times V$), which we call the \emph{all pairs effective resistance estimation problem}. In Section~\ref{sec:improved_lower_bounds_spectral_sums} we show that our techniques also yield lower bounds for other spectral sum estimation problems.
\subsection{Our Approach}\label{sec:lb_overview}
\paragraph{Approaches of Prior Work}
The approach of \cite{musco2017spectrum} begins with the fact that $G$ has a triangle if and only if $\trace{\A_G^3}/6 \geq 1$. They use the fact that various spectral sums $\mathcal{S}_f$ of the SDD matrix $\I - \delta \A_G$ (for $\delta$ sufficiently small) can be expressed as a power series $\mathcal{S}_f(\I - \delta \A_G) = \sum_{k=0}^\infty c_k \delta^k \tr(\A_G^k)$.
The first two terms of this series can be computed directly. So given $Y \approx_\epsilon \mathcal{S}_f(\I - \delta \A_G)$, one can estimate $\trace{\A_G^3}$, where the estimation error is controlled by the magnitude of the first two terms of the series and the tail error due to truncating at the third term. By bounding this estimation error, \cite{musco2017spectrum} shows that, for appropriate choices of $\delta$, $Y \approx_\epsilon \mathcal{S}_f(\I - \delta \A_G)$ yields an additive 1/2 approximation to $\trace{\A_G^3}$, which is sufficient for triangle detection. They also reduce the problem of estimating the spectral sum $\trace{\B^{-1}}$ for an SDD matrix $\B$ to the all pairs effective resistance estimation problem. 

\paragraph{Our Approach} We use three key techniques to better bound the estimation errors incurred in the power-series-inspired approach of \cite{musco2017spectrum}. This yields faster reductions and better lower bounds for effective resistance estimation. Rather than obtaining effective resistance lower bounds by reducing the problem of computing $\trace{\A_G^3}/6$ to computing the \emph{trace} of an SDD matrix as in \cite{musco2017spectrum}, we use a reduction that closely leverages the structure of effective resistances and properties of $\A_G^2$. 

We begin with the fact that for $\alpha>0$ sufficiently small, 
\begin{align}\label{eq:series2}
    \left(\I - \frac{\alpha}{n}\A_G\right)^{-1} = \sum_{k=0}^\infty \frac{\alpha^k}{n^k} \A_G^k. 
\end{align}
Since $\A_G$ is known, given access to $\deltaij^\top (\I - \frac{\alpha}{n}\A_G)^{-1} \deltaij$, we can estimate the entries of $\A_G^2$, where the estimation error is controlled by $\alpha$ and the tail error of truncating \eqref{eq:series2} at the third term. By bounding this estimation error, for appropriate choice of $\alpha$, we can obtain additive 1/2 approximations to all entries of $\A_G^2$, which is sufficient to identify all paths of length 2. We can then detect a triangle by simply scanning for an edge $\{u, v\}$ such that $u$ and $v$ are connected by a path of length 2. Estimating the entries of $\A_G^2$ leads to lower estimation error than estimating $\trace{\A_G^3}$ as in \cite{musco2017spectrum}).

Second, we use a standard randomized reduction that reduces the triangle detection problem to the triangle detection problem restricted to tripartite graphs. The reduction relies on the fact that a randomly sampled tripartition of the original graph preserves triangles with constant probability. To detect a triangle in a tripartite graph $G = (V_1 \sqcup V_2 \sqcup V_3, E)$, we construct a graph $\Gprime$ by removing all edges $E_{1,2} := \{\{u, v\} \in E : u \in V_1, v \in V_2\}$ between $V_1$ and $V_2$. $G$ has a triangle if and only if there is an edge $\curs{u,v} \in E_{1,2}$ and a path of length 2 between $u$ and $v$ in $\Gprime$. Crucially, we can show that the third term $\As_{\Gprime}^3$ does not contribute to the tail error when estimating the $\{u, v\}$-th entry of $\pars{\A_{\Gprime}^2}$ using \eqref{eq:series2}.

Third, to lower the spectral norm of $\A_{\Gprime}$ (and consequently better bound the convergence of the power series \eqref{eq:series2}), we use a symmetric random signing of $\A_{\Gprime}$ defined below.
\begin{definition}[Symmetric Random Signing]
    Given a symmetric matrix $\A \in \R^{n \times n}$, its \emph{symmetric random signing} $\As$ is the random matrix with $\As_{i, j} := \rad_{i,j} \A_{i,j}$, where $\rad_{i,j}$ are independent Rademacher random variables (i.e. $\rad_{i,j} = \pm1$, each with probability 1/2) that satisfy $\rad_{i,j} = \rad_{j,i}$.
\end{definition}
We show that with constant probability, this random signing preserves whether the entries of $\A_{\Gprime}^2$ are non-zero, allowing us to detect if $G$ has a triangle even if we replace $\A_G$ in \eqref{eq:series2} with $\As_{\Gprime}$ instead. This is beneficial, as matrix Chernoff guarantees $\Norm{\As_{\Gprime}}_2 = \Otilde{\sqrt{n}}$ whp. whereas $\Norm{\A_{\Gprime}}_2$ may be as large as $n$. This means that by replacing $\A_G$ with $\As_{\Gprime}$, the norm of each term in \eqref{eq:series2} decreases from $O(\alpha^k)$ to $O(\alpha^k n^{-k/2})$. So the tail error of truncating the power series is smaller. 

To compute entries of $\As_{\Gprime}^2$ efficiently using effective resistance estimates on expanders, we first show that we can use all pairs effective resistance estimates on expanders to estimate $\deltaij^T \M^{-1} \deltaij$ for all $i,j \in [n]$, where $\M = (\I - \Q)$ is an SDD matrix with $\rho(\Q) \leq 1/3$. Then, by choosing $\M = \I - \frac{\alpha}{n} \As_{\Gprime}$ as in \eqref{eq:series2} for an appropriate constant $\alpha$, we can estimate $\As_{\Gprime}^2$ from estimates of $\deltaij^T \M^{-1} \deltaij$. This yields our lower bound on the all pairs effective resistance estimation problem.

Additionally, we show that with constant probability, the random signing preserves the property that $\trace{\A_G^3}$ is non-zero. We leverage this aspect of the random signing to obtain improved randomized conditional lower bounds for various spectral sum estimation problems.  We closely follow the trace estimation approach of \cite{musco2017spectrum}, and again use the smaller spectral radius of $\As_G$ to improve bounds on the power series truncation error.

\subsection{Improved Lower Bounds for Effective Resistance Estimation}\label{sec:improved_lower_bounds_effres}
In this section we provide a series of reductions which yield our main result on lower bounds for the all pairs effective resistance estimation problem for all pairs of nodes (case where $S = V \times V$). First, we formalize a standard randomized reduction from triangle detection in general graphs to triangle detection in tripartite graphs in \Cref{lemma:tripartite_to_normal_tridetect} below.

\begin{restatable}{lemma}{lemmaTriangleToTripartite}\label{lemma:tripartite_to_normal_tridetect}
Given an algorithm which can solve the triangle detection problem on an $n$-node tripartite undirected graph in $O(n^\gamma)$ time, we can produce a randomized algorithm which can solve the triangle detection problem on an arbitrary $n$-node undirected graph $G$ in $\widetilde{O}(n^\gamma)$ time whp.
\end{restatable}
\begin{proof}
    We first sample a tripartite subgraph $H$ of $G$ by assigning each vertex in $G$ to a random tripartition with equal probability 1/3 and deleting edges within each resulting tripartition. First, note that if $G$ has no triangles then $H$ also has no triangles since $H$ is a subgraph of $G$. Second, observe that if $G$ has a triangle $\curs{a,b,c}$, it is also a triangle in $H$ if each vertex ends up in a different tripartition. This occurs with probability at least $(1/3)^3 = 1/27$ which is a constant. Therefore, solving the triangle detection problem on $H$ and returning the same output also successfully solves the triangle detection problem on $G$ with probability at least $1/27$. We can repeat this randomized procedure $\log (n^c)$ times to boost the success probability to at least $1 - n^{-c}$, which is whp. in $n$. This randomized algorithm runs in $\widetilde{O}(n^\gamma)$ time, completing the proof.
\end{proof}

\begin{definition}\label{def:Reff_SDD} In the \emph{all pairs SDD effective resistance estimation problem}, given an SDD matrix $\M$ such that $\diagonal{\M} = \I$, $\adjacency{\M} = \Q$, with $\rho(\Q) \leq 1/3$ and $\epsilon \in (0, 1)$, we must output $\widetilde{r} \in \R^{n^2}$ such that $\widetilde{r}_{a,b} \approx_\epsilon \deltaab^\top \M^{-1}\deltaab$ $\forall a, b \in [n]$. We call $\deltaab^\top \M^{-1}\deltaab$ the SDD effective resistance of $(a, b)$ in $\M$. 
\end{definition}

For brevity, we use $\widetilde{r}(\M)$ to refer to the solution of the SDD effective resistance problem on input $\M$. Our first step is to show that an algorithm for the all pairs effective resistance estimation problem on expanders implies an algorithm for the all pairs SDD effective resistance estimation problem.

\begin{restatable}{lemma}{graphEffresToSDDEffres}\label{lemma:grapheffres_to_sddeffres} 
Given an algorithm to solve the all pairs effective resistance estimation problem on graphs with $\widetilde{\Omega}(1)$-expansion in $\Otilde{n^2\epsilon^{-c}}$ time for some $c > 0$, we can produce an algorithm to solve the all pairs SDD effective resistance estimation problem in $\Otilde{n^2 \epsilon^{-c}}$ time. 
\end{restatable}

To prove Lemma~\ref{lemma:grapheffres_to_sddeffres}, we first prove the lemma for the case where $\Q$ is entrywise non-negative in \Cref{lemma:lowerbound_thm2_lemma1} by constructing an expander $G$ with $n+1$ vertices such that $r_{G}(a, b) = \deltaab^\top \M^{-1}\deltaab$ for all $a, b \in [n]$. Then, in \Cref{lemma:lowerbound_thm2_lemma2}, we extend the reduction to arbitrary $\Q$ by constructing $\Q'$ of size $2n$ so that $\Q'$ is entrywise non-negative and $\widetilde{r}(\I - \Q)$ is a simple linear transform of $\widetilde{r}(\I - \Q')$. 

\begin{lemma}\label{lemma:lowerbound_thm2_lemma1} Given an algorithm that solves the all pairs effective resistance estimation problem on graphs with $\widetilde{\Omega}(1)$-expansion in $\Otilde{n^2\epsilon^{-c}}$ time for some $c > 0$, we can produce an algorithm which takes as input an SDD matrix $\M = \I - \Q$ such that $\Q$ is entrywise non-negative and $\Rho{\Q} \leq \frac{1}{3}$, and solves the all pairs SDD effective resistance estimation problem for $\M$ in $\Otilde{n^2 \epsilon^{-c}}$ time.
\end{lemma}

\begin{proof}
Let $v := (\I-\Q) \one$. Note that $v$ is entrywise non-negative. Consider the matrix 
\begin{align*}
    \mathbf{L} := \begin{pmatrix}
        \I & 0 \\
        0 & \Norm{v}_1
    \end{pmatrix} - \begin{pmatrix}
        \Q & v \\
        v^\top & 0
    \end{pmatrix}. 
\end{align*}
and note that it is the Laplacian matrix. Since $\M$ is a principal submatrix of $\mathbf{L}$, we have by the eigenvalue interlacing theorem \cite{hwang2004cauchy} that $\lambda_1 \pars{\M} \leq \lambda_2 \pars{ \mathbf{L} }$. But since $\rho\pars{\Q} \leq 1/3$, we have that $\lambda_1 \pars{\M} \leq 2/3$ and consequently  $\lambda_2\pars{\mathbf{L}} \geq 2/3.$ Therefore, $\mathbf{L}$ is the Laplacian of an $\widetilde{\Omega}(1)$-expander (see Cheeger's Inequality \Cref{thm:cheeger}).

We claim that for $i,j \in n$, 
\begin{align*}
    \deltaij^\top \Ldagger \deltaij = \deltaij^\top \M^{-1} \deltaij, 
\end{align*}
which is sufficient to prove the lemma. Note that for any $x \in \R^n$, $y \in \R$, and $\alpha \in \R$, 
\begin{align*}
    \mathbf{L} \begin{pmatrix}
        x + \alpha \one \\
        y + \alpha
    \end{pmatrix} = \begin{pmatrix}
        (\I - \Q) x + v^\top x \\
        v^\top x - \Norm{v}_1 y 
    \end{pmatrix}. 
\end{align*}
Let 
\begin{align*}
    \Ldagger \deltaij = \begin{pmatrix}
        z_x \\
        z_y
    \end{pmatrix}. 
\end{align*}
Then, note that 
\begin{align*}
    \mathbf{L} \begin{pmatrix}
        z_x - z_y \one \\
        0
    \end{pmatrix} = \begin{pmatrix}
        \deltaij \\
        0
    \end{pmatrix} \implies (\I - \Q) (z_x - z_y \one) = \deltaij. 
\end{align*}

Consequently, 
\begin{align*}
    \deltaij^\top (\I - \Q)^{-1} \deltaij = \deltaij^\top (z_x - z_y \one) = \deltaij^\top z_x - z_y \deltaij^\top \one = \deltaij^\top z_x = \deltaij^\top \Ldagger \deltaij. 
\end{align*}
\end{proof}

Now, to prove \Cref{lemma:lowerbound_thm2_lemma2}, we first show a useful property of block symmetric matrices in the helper lemma below.
\begin{lemma}\label{lemma:lowerbound_thm2_lemma0} Suppose $\A = \begin{pmatrix}
    \X & \Y \\
    \Y & \X
\end{pmatrix}$ where $\X, \Y \in \mathbb{R}^{n \times n}$. Then, 
\begin{align*}
    \deltaij^\top (\X-\Y) \vec{\delta}_{i,j} = \frac{1}{2} \left[\vec{\delta}_{i,n+i}^\top \A \vec{\delta}_{i,n+i} + \vec{\delta}_{j,n+j}^\top \A \vec{\delta}_{j,n+j}\right] - \vec{\delta}_{i,n+j}^\top \A \vec{\delta}_{i,n+j} + \vec{\delta}_{i,j}^\top \A \vec{\delta}_{i,j}. 
\end{align*}
\end{lemma}
\begin{proof} Note that 
\begin{align*}
    \deltaij^\top (\X-\Y) \vec{\delta}_{i,j} &= (\X_{i,i} + \X_{j,j}) - 2 \X_{i,j} + 2\Y_{i,j} - (\Y_{i,i} + \Y_{j,j}). 
\end{align*}
Meanwhile, 
\begin{align*}
    \frac{1}{2} \vec{\delta}_{i,n+i}^\top \A \vec{\delta}_{i,n+i} &= (\X_{i,i} - \Y_{i,i}), \\
    \frac{1}{2} \vec{\delta}_{j,n+j}^\top \A \vec{\delta}_{j,n+j} &= (\X_{j,j} - \Y_{j,j}), \\
    -\vec{\delta}_{i, n+j}^\top \A \vec{\delta}_{i,n+j} &= -\X_{i,i} + 2 \Y_{i,j} - \X_{j,j}, \\
    \deltaij^\top \A \deltaij &= \X_{i,i} - 2 \X_{i,j} + \X_{j,j}, 
\end{align*}
and adding these four terms together concludes the proof. 
\end{proof}

\begin{lemma}\label{lemma:lowerbound_thm2_lemma2} Suppose we are given an algorithm which takes as input an SDD matrix $\M = \I - \Q$ such that $\Q$ is entrywise non-negative and $\Rho{\Q} \leq \frac{1}{3}$, and solves the all pairs SDD effective resistance estimation problem for $\M$ in $\Otilde{n^2 \epsilon^{-c}}$ time. Then, we can produce an algorithm which takes as input an SDD matrix $\Mp = \I - \Qp$ and $\Rho{\Qp} \leq \frac{1}{3}$, and solves the all pairs SDD effective resistance estimation problem for $\Mp$ in $\Otilde{n^2 \epsilon^{-c}}$ time.
\end{lemma}

\begin{proof} We can decompose $\Qp$ as $\Qp = \Ps - \N$ where $\Ps$ is a matrix which contains only the positive offdiagonal entries of $\Qp$ and $-\N$ is a matrix which contains all the negative offdiagonal entries. Therefore both $\Ps$ and $\N$ themselves are entrywise non-negative. We define
\begin{align*}
    \Q :=  \begin{pmatrix}
        \Ps & \N \\
        \N & \Ps
    \end{pmatrix}.
\end{align*}
 
Note that $\Q$ is also entrywise non-negative.

We also have $\rho (\Q) \leq 1/3$. To see this, assume for the sake of contradiction that $\Q$ has an eigenvalue $\lambda > 1/3$. This means that there must exist some $x \in \R^{2n}$ such that $\Q x = \lambda x$. Let $x = [x_1; x_2]$ where $x_1, x_2 \in \mathbb{R}^n$. The eigenvalue equation then implies that $\Ps x_1 + \N x_2 = \lambda x_1$ and $\N x_1 + \Ps x_2 = \lambda x_2$. Subtracting these equations yields $(\Ps - \N)(x_1 - x_2) = \lambda (x_1 - x_2)$ which means that $\lambda$ is also an eigenvalue of $\Qp$. This is a contradiction since $\rho(\Qp) \leq 1/3$. Therefore, $\rho(\Q) \leq 1/3$.

Now, consider the following block decomposition of $\Q^k$
\begin{align*}
    \Q^k = \begin{pmatrix}
        \X & \Y \\ 
        \Y & \X
    \end{pmatrix}
\end{align*}
where $\X, \Y \in \R^{n \times n}$. We will show by induction that $\Qp^k = (\Ps - \N)^k = \X - \Y$. In the base case, when $k = 1$, this is trivially true. Now, assume that the claim holds for all $k \leq m$ for some $m$. Consider the following block decomposition of $\Q^m$
\begin{align*}
    \Q^m = \begin{pmatrix}
        \W & \Z \\ 
        \Z & \W
    \end{pmatrix}.
\end{align*}
By the inductive hypothesis, we know that  $\Qp^m = (\Ps - \N)^m = \W - \Z$. Now, we have that
\begin{align*}
    \Q^{m+1} = \begin{pmatrix}
        \W & \Z \\ 
        \Z & \W
    \end{pmatrix}  \begin{pmatrix}
        \Ps & \N \\ 
        \N & \Ps
    \end{pmatrix} = \begin{pmatrix}
        \W\Ps + \Z\N & \W\N + \Z\Ps \\ 
        \W\N + \Z\Ps & \W\Ps + \Z\N
    \end{pmatrix},
\end{align*}
and also that 
\begin{align*}
    \Qp^{m+1} = (\W - \Z)(\Ps - \N) = \W\Ps - \Z\Ps - \W\N + \Z\N = (\W\Ps + \Z\N) - (\W\N + \Z\Ps). 
\end{align*}
Hence, the claim follows by induction. By Lemma~\ref{lemma:lowerbound_thm2_lemma0}, it follows that 
\begin{align}\label{eq:lemma:lowerbound_thm2_lemma1_eq1}
    \deltaij^\top \Qp^k \deltaij =  \frac{1}{2} \left[\vec{\delta}_{i,n+i}^\top \Q^k \vec{\delta}_{i,n+i} + \vec{\delta}_{j,n+j}^\top \Q^k \vec{\delta}_{j,n+j}\right] - \vec{\delta}_{i,n+j}^\top \Q^k \vec{\delta}_{i,n+j} + \vec{\delta}_{i,j}^\top \Q^k \vec{\delta}_{i,j}. 
\end{align}

Now we can use the power series expansion of $(\I - \Q)^{-1}$ to say, for any $u, v \in [2n]$,
\begin{align*}
    \vec{\delta}_{u,v}^\top (\I - \Q)^{-1} \vec{\delta}_{u,v} = \sum_{k=0}^\infty \vec{\delta}_{u,v}^\top \Q^k \vec{\delta}_{u,v}. 
\end{align*}
Similarly, for any $i, j \in [n]$, 
\begin{align*}
    \vec{\delta}_{i,j}^\top (\I - \Qp)^{-1} \vec{\delta}_{i,j} = \sum_{k=0}^\infty \vec{\delta}_{i,j}^\top \Qp^k \vec{\delta}_{i,j}. 
\end{align*}
So, by linearity and \eqref{eq:lemma:lowerbound_thm2_lemma1_eq1}, it follows that 
\begin{align*}
    \deltaij^\top (\I - \Qp)^{-1} \deltaij =&  \frac{1}{2} \left[\vec{\delta}_{i,n+i}^\top (\I - \Q)^{-1} \vec{\delta}_{i,n+i} + \vec{\delta}_{j,n+j}^\top (\I - \Q)^{-1} \vec{\delta}_{j,n+j}\right]\\ 
    &-\vec{\delta}_{i,n+j}^\top (\I - \Q)^{-1} \vec{\delta}_{i,n+j} + \vec{\delta}_{i,j}^\top (\I - \Q)^{-1} \vec{\delta}_{i,j},
\end{align*}
and this completes the proof.

\end{proof}
\Cref{lemma:grapheffres_to_sddeffres} then follows directly from \Cref{lemma:lowerbound_thm2_lemma1} and \Cref{lemma:lowerbound_thm2_lemma2}. We now turn our attention to reducing the triangle detection problem to the all pairs SDD effective resistance estimation problem. As discussed, a key aspect of our approach is to work with the random signing $\widebar{\A}_G$ of $\A_G$. Lemma~\ref{lemma:signing_2paths} shows that to determine whether $(\A_G^2)_{i,j} > 0$ with constant probability, it suffices to determine whether $(\widebar{\A}_G^2)_{i,j} > 0$.

\begin{restatable}{lemma}{lemmaSigningTwoPaths}\label{lemma:signing_2paths} For $i \neq j$, if $(\A_{G}^2)_{i,j} = 0$, then $(\As_{G}^2)_{i,j} = 0$; if $(\A_G^2)_{i,j} > 0$, $\mathbb{P}\left[\lvert (\As_G^2)_{i,j}\rvert > 1\right] \geq 1/2$. 
\end{restatable}
\begin{proof} 
The idea of the proof is that if $\{a, b\}, \{b, c\}$ exist in $G$, either $\xi_{a,c} = 1$ or $\xi_{a, c} = -1$ results in $(\As_{G}^2)_{a,c} > 0$. Note that 
\begin{align*}
    \left(\As^2\right)_{i,j} = \sum_{k=1}^n \As_{i,k} \As_{k,j}
\end{align*}

If $\A^2_{i,j} = 0$, then $G$ has no path of length exactly two between $i$ and $j$, so each term $\A_{i,k} \A_{k,j}$ in the summation above must be zero, and hence $\left(\As^2\right)_{i,j} = 0$, completing the first part of the lemma.

Since $\left(\Abar^2\right)_{i,j}$ is only supported on the integers, to prove the second statement, it suffices to show that when $\left(\A^2\right)_{i,j} > 0$, $\mathbb{P}\left[\left\lvert \left(\As^2\right)_{i,j}\right\rvert = 0\right]$ is no larger than 1/2. To see this, note that if $\A^2_{i,j} \neq 0$, then then $G$ has at least one path of length exactly 2 between $i$ and $j$. That is, there exists a $k' \in [n]$ such that $\A_{i,k'} \A_{k', j} = 1$. We can write 
\begin{align*}
    \left(\As^2\right)_{i,j} = \sum_{k = 1}^n \xi_{i, k} \xi_{k,j} \A_{i,k} \A_{k,j}
\end{align*}
Because $i \neq j$, for every $k, \ell \in [n]$ each $\xi_{i,k}$ is always independent from any other $\xi_{\ell, j}$ term appearing in the sum. Moreover, if $\xi_{i,k}$ appears in the sum, then $\xi_{k,i}$ never appears in the sum. Therefore, for each $\{\xi_{i, k} \xi_{k,j}\}_{k\in[n]}$ are themselves independently chosen Rademacher random variables, and all terms in the summation are independent. Separating out the $k'$-th entry,
\begin{align*}
    \left(\As^2\right)_{i,j} = \xi_{i, k'} \xi_{k', j} + \sum_{k \neq k'}^n \xi_{i, k} \xi_{k,j} \A_{i,k} \A_{k,j} \overset{D}{=} \xi' + Z, 
\end{align*}
where $\xi'$ is a Rademacher random variable and $Z := \sum_{k \neq k'}^n \xi_{i, k} \xi_{k,j} \A_{i,k} \A_{k,j}$ is independent from $\xi'$. For any value of $Z$, $\mathbb{P}[\xi' = -Z] \leq \frac{1}{2}$. Therefore, 
\begin{align*}
    \mathbb{P}\left[\left(\As^2\right)_{i,j} = 0\right] \leq \frac{1}{2}, 
\end{align*}
completing the second part of the result. 
\end{proof}

Matrix Chernoff (\Cref{lemma:concentration}) ensures whp. $\rho({\As}_G) = \Otilde{\sqrt{n}}$, while $\rho({\A}_G)$ could be as large as $n$ \cite{tropp2012user}. So, estimating entries of ${\As}_G$ leads to lower power series tail error. 

\begin{restatable}{lemma}{lemmaConcentration}\label{lemma:concentration} Let $G = (V, E)$ be an undirected unweighted graph on $n$ nodes. Let $\As_G$ denote a symmetric random signing of $\A_G$. With high probability, $\rho(\As) \leq \Otilde{\sqrt{n}}$.  
\end{restatable}
\begin{proof} Let $\sigma^2 := \Norm{\sum_{\{u,v\} \in E} \left(\mathbf{E}_{u,v}\right)^2 }_2$ where $\mathbf{E}_{u,v}$ is the adjacency matrix of a graph with only a single edge between $u$ and $v$. Note that entries of $\mathbf{E}_{u,v}^2$ indicates paths of length two in this graph. Therefore, this is a diagonal matrix that satisfies $\left(\mathbf{E}_{u,v}^2\right)_{i,i} = 1$ if and only if $i \in \{u, v\}$. Consequently, 
\begin{align*}
    \sigma^2 = \Norm{\sum_{\{u,v\} \in E} \left(\mathbf{E}_{u,v}\right)^2}_2 = \Norm{\D}_2 \leq \dmax \leq n. 
\end{align*}
We can now write $\As = \sum_{\{u,v\} \in E} \rad_{u,v} \mathbf{E}_{u,v}$ and apply the Matrix Rademacher concentration result (Theorem 1.2) from \cite{tropp2012user}, to get that for any constant $c > 1$, 
\begin{align*}
    \Prob{\lambda_{n}\left( \Abar \right) \geq \sqrt{\dmax} \log(c n)} \leq n \Exp{\frac{-c \dmax \log(n)}{2\dmax}} = n { \frac{1}{n^c} } = n^{-c+1}. 
\end{align*}
\end{proof}

Finally, we use the power series approach in Section~\ref{sec:lb_overview} to obtain \Cref{theorem:sddeffres_hard_as_tridetect}.
\begin{restatable}{theorem}{thmSDDeffresToTripartiteTriangle} \label{theorem:sddeffres_hard_as_tridetect} Given an algorithm which solves the all pairs SDD effective resistance estimation problem in $\Otilde{n^2 \epsilon^{-c}}$ time, we can produce a randomized algorithm that solves the triangle detection problem in $\Otilde{n^{2(1+c)}}$ time whp. 
\end{restatable}
\begin{proof} As $G$ is tripartite, let $V = V_1 \sqcup V_2 \sqcup V_3$ be the partition of $G$ such that no edge has both endpoints in $V_i$ for some $i \in [3]$. Let $E_{1,2} := \{\{u, v\} \in E : u \in V_1, v \in V_2\}$, and let $\Gprime := (V, E\setminus{E_{1,2}})$. Let $\A = \A_{\Gprime}$ denote the adjacency matrix of $\Gprime$ and let $\As$ be a random signing of $\A$.

Suppose that $G$ has a triangle. Then, there exists a pair of vertices $u \in V_1, v \in V_2$ such that $\{u, v\} \in E_{12}$ and $\Gprime$ contains a path of length two between $u$ and $v$. Furthermore, observe that because there are no edges between $V_1$ and $V_2$ in $\Gprime$, $\Gprime$ has no paths of length three between $V_1$ and $V_2$. Consequently, in order to find a triangle in $G$, it suffices to check, for each $i \in V_1, j \in V_2$ with $\{i, j\} \in E_{1,2}$, whether there exists a path of length two between node $i$ and node $j$ in $\Gprime$. In other words, we need to check if there exists some $\curs{i, j} \in E_{1,2}$ such that $\A^2_{i,j} > 0$. By Lemma~\ref{lemma:signing_2paths}, we can instead check if $\abs{\As^2_{i,j}} > 0$ and we would still correctly detect a triangle with probability at least 1/2. Note that this check requires requires only $O(\nnz(\A))$ additional time, since $\Abs{E_{1,2}} < \nnz(\A)$.

So, our goal now is to compute an accurate enough estimate of $\abs{\As_{i,j}}^2$ for all $\curs{i,j} \in E_{1,2}$ given the effective resistance estimate $ \rtildeij$. To this end, let $\N = \left(\I - \frac{\alpha}{n}\Abar \right)^{-1}$ for some $\alpha < \frac{1}{3}$. Note that the max row-sum of $\Abar$ is 1/3, so the inverse exists. Lemma~\ref{lemma:concentration} guarantees that with high probability, $\Rho{\Abar} \leq \Otilde{\sqrt{\dmax}} = \Otilde{\sqrt{n}}$. We condition on this event in the remainder of the proof. Consequently, we can express $\N$ as a power series,
\begin{align*}
    \N = \sum_{k=0}^\infty \left(\frac{\alpha}{n}\right)^k \Abar^k. 
\end{align*}
Now, let $\Ntilde$ denote the truncation of $\N$ at the third term in this power series. That is, $\Ntilde = \I + \frac{\alpha}{n} \Abar + \frac{\alpha^2}{n^2} \Abar^2$. Therefore, we have for all $\curs{i,j} \in E_{1, 2}$, 
\begin{align*}
    \As_{i,j}^2 = \frac{n^2}{\alpha^2}\Ntilde_{i, j}.
\end{align*}
Noticing that $\N_{i,j} = \frac{\N_{i,i} + \N_{j,j} - r_{i,j}}{2} $ 
 motivates us to define our estimate of $\As^2_{i,j}$ that we denote $P_{i,j}$ as follows
\begin{align*}
    P_{i,j} := \frac{n^2}{\alpha^2} \left[ \frac{\Ntilde_{i,i} + \Ntilde_{j,j} - \rtildeij}{2} \right].
\end{align*}
Now, observe that $\Ntilde_{i,i} = 1 + \frac{\alpha^2}{n^2}\As^2_{i,i}$, and $\As^2_{i,i}$ is simply the degree of vertex $i$ in $\Gprime$. Note that the random signing does not affect the fact that the diagonal entries of the square of the adjacency matrix are the degrees. Therefore, $\Ntilde_{i,i}$ can also be computed for all $i$ in only $O(\nnz(\A))$ time. The additive error between our estimate $P_{i,j}$ and $\abs{\As}^2_{i,j}$ takes the form
\begin{align*}
    \abs{P_{i,j} - \As^2_{i,j}} = \abs{\frac{n^2}{\alpha^2}\sqrs{\Ntilde_{i,j} - \frac{\Ntilde_{i,i} + \Ntilde_{j,j} - \rtildeij}{2}}}.
\end{align*}
By triangle inequality and plugging in the definition of $r_{i,j}$, we break up the error into four pieces
\begin{align}\label{eq:error_in_four_pieces}
    \abs{P_{i,j} - \As^2_{i,j}} \leq \frac{n^2}{\alpha^2} \sqrs{\abs{\Ntilde_{i, j} - \N_{i, j}} + \abs{\Ntilde_{i, i} - \N_{i, i}}/2 + \abs{\Ntilde_{j, j} - \N_{j, j}}/2 +\abs{\rtildeij - r_{i, j}}/2 }.
\end{align}
We now bound each term separately. Consider any $i \in V_1$. Since, $\Gprime$ contains no triangle containing $i$, $\left(\Abar^3\right)_{i,i} = 0$. So, we have that 
\begin{align*}
    \left\lvert \N_{i,i} - \Ntilde_{i,i}\right\rvert &= \left\lvert \sum_{k=4}^\infty \frac{\alpha^k}{n^k} \left(\Abar^k\right)_{i,i}\right\rvert \leq \Norm{ \sum_{k=4}^\infty \frac{\alpha^k}{n^k} \Abar^k }_2 \\
    &\leq \sum_{k=4}^\infty \widetilde{O}\pars{\frac{\alpha}{\sqrt{n}}}^k = \frac{\Otilde{\frac{\alpha}{\sqrt{n}}}^4}{1 - \Otilde{\frac{\alpha}{\sqrt{n}}}}
\end{align*}
Similarly, for any $i \in V_1$ and $j \in V_2$, note that $\Abar^3_{i,j} = 0$ because $\Gprime$ has no edges between $V_1$ and $V_2$, and there are no edges within each tripartition $V_i$. Consequently, by a similar argument as above, 
\begin{align*}
    \left\lvert \N_{i,j} - \Ntilde_{i,j}\right\rvert &= \sum_{k=4}^\infty (\Abar^4)_{i,j} = \frac{1}{2}\sum_{k=4}^\infty e_i^\top (\Abar^k) e_i  +  e_j^\top (\Abar^k) e_j  - \vec{\delta}_{i,j}^\top (\Abar^k) \vec{\delta}_{i,j} \\
    &= \frac{1}{2} \left( e_i^\top e_i + e_j^\top e_j + \vec{\delta}_{i,j}^\top \vec{\delta}_{i,j}\right) \Norm{ \sum_{k=4}^\infty \frac{\alpha^k}{n^k} \Abar^k }_2 \\
    &\leq 2 \sum_{k=4}^\infty \widetilde{O}\pars{\frac{\alpha}{\sqrt{n}}}^k = \frac{2 \Otilde{\frac{\alpha}{\sqrt{n}}}^4}{1 - \Otilde{\frac{\alpha}{\sqrt{n}}}}. 
\end{align*} 

Finally, consider the magnitude of the approximation error between $\rtildeij$ and $r_{i,j}$. We have
\begin{align*}
    \abs{r_{i,j} - \rtildeij} \leq \epsilon \abs{r_{i,j}}.
\end{align*}
Note that $\abs{r_{i,j}} = \abs{\deltaij^\top \N \deltaij} \leq \norm{\deltaij}^2 \norm{\N} \leq 2 \norm{\pars{\I - \frac{\alpha}{n}\As}^{-1}} \leq \frac{2}{1-\frac{\alpha}{n} \Norm{\Abar}_2 } \leq \frac{2}{1 - \Otilde{\frac{\alpha}{\sqrt{n}}}}. $
Plugging these estimates of the errors into \eqref{eq:error_in_four_pieces}, we get
\begin{align*}
    \abs{P_{i,j} - \As^2_{i,j}} \leq \frac{n^2}{\alpha^2} \sqrs{\frac{\widetilde{O}\pars{\frac{\alpha}{\sqrt{n}}}^4 + \epsilon}{1 - \widetilde{O}\pars{\frac{\alpha}{\sqrt{n}}}}}.
\end{align*}

Since $\Abs{\Abar^2_{i,j}} \in \{0, 1\}$, to compute $\Abar^2_{i,j}$, we need approximate it to additive 1/2 error. That is, we require 
\begin{align*}
    \frac{1}{1 - \Otilde{\frac{\alpha}{\sqrt{n}}}} \left[{\Otilde{\alpha^2}} + \epsilon \frac{n^2}{\alpha^2} \right] &< \frac{1}{2},
\end{align*}
or equivalently, 
\begin{align*}
    \epsilon &< \frac{{\alpha^2 - \Otilde{\frac{\alpha^3}{\sqrt{n}}}}}{2n^2} - \frac{\Otilde{\alpha^4}}{n^2} = \widetilde{O}\pars{\frac{1}{n^2}},
\end{align*}
where the last step follows from the fact that we can take $\alpha$ to be a sufficiently small constant. Therefore, estimates $\rtildeij$ with $\epsilon = \widetilde{O}(n^{-2})$ for all $\curs{i, j} \in E_{1,2}$ are sufficient to determine if $\Abs{\Abar^2_{i,j}}$ is 0 or 1. As noted earlier, checking this for all edges in $E_{1,2}$ takes $O(\nnz (\A))$ additional time. Therefore, by plugging in $\epsilon = \widetilde{O}(n^{-2})$, we can use an algorithm that solves the all pairs SDD effective resistance estimation problem in $\widetilde{O}(n^2 \epsilon^{-c})$ time to solve the triangle detection problem in $\widetilde{O}(n^2 n^{2c}) = \widetilde{O}(n^{2(1+c)})$ time whp. and this completes the proof.  
\end{proof}

\Cref{theorem:sddeffres_hard_as_tridetect} and \Cref{lemma:grapheffres_to_sddeffres} with $c = 1/2 - \delta$ immediately imply our main result \Cref{thm:reff_lb_main_result}.

\subsection{Improved Lower Bounds for Spectral Sum Estimation} \label{sec:improved_lower_bounds_spectral_sums}
Finally, we discuss our improved lower bounds for various spectral sum estimation problems. Analogous to Lemma~\ref{lemma:signing_2paths}, in the following lemma we show that to determine whether a graph has a triangle (i.e., $\trace{\A_{G}^3} > 0$) with constant probability, it suffices to determine whether $\trace{\As_{G}^3} > 0$.

\begin{restatable}{lemma}{lemmaRandomSigningPreservesTriangles}\label{lemma:random_signing_preserves_triangles}
    If $\trace{\A_G^3} = 0$, then $\trace{\As_G^3} = 0$, and if $\trace{\A_G^3} > 0$ then $\Prob{\abs{\trace{\As_G^3}} > 0} \geq 1/4$. 
\end{restatable}

\begin{proof}
    The central idea of the proof is that if a triangle $\{a, b\}, \{b, c\}, \{c, a\}$ exists in $G$, then amongst the 4 possible configurations of the Rademacher random signing variables $\xi_{a,b}$ and $\xi_{b,c}$, at least one configuration must result in $\abs{\trace{\As_G^3}} > 0$. 

    Again for convenience, let $\As = \As_G$. Denote by $\rad_{ab} = \rad_{ba}$ the Rademacher random variable used to decide the sign of edge $(a,b)$. We can write
    \begin{align}\label{eq:sum_over_all_triangles}
        \trace{\As^3}/6 = \sum_{\text{triangles } \curs{i,j,k} \text{ in } G}  \rad_{ij} \rad_{jk} \rad_{ki} =: T
    \end{align}
    If $\trace{\A^3} > 0$, then $G$ must have at least one triangle. Consider the following cases:
    \begin{enumerate}
        \item $G$ has an odd number of triangles. In this case, since each term in the sum in \eqref{eq:sum_over_all_triangles} is either $+1$ or $-1$, and there is an odd number of terms in the sum, so $\trace{\As^3} > 0$ wp $1$.
        \item $G$ has an even number of triangles. First, define 
        \begin{align*}
            T(a, b) := \sum_{\text{triangles } \curs{a,b,k} \text{ that contain }\curs{a,b}} \rad_{ab} \rad_{bk} \rad_{ak}
        \end{align*}
        We now subdivide this into two cases:
        \begin{enumerate}
            \item There exists an edge $\curs{a,b}$ that is part of an odd number of triangles. We can decompose the sum in \eqref{eq:sum_over_all_triangles} as follows:
            \begin{align*}
                \trace{\As^3}/6 = \underbrace{T(a,b)}_{S_1} + \underbrace{T-T(a,b)}_{S_2}
            \end{align*}
            Suppose there exists some realization of the random variables $\rad_{ij}$ such that $\trace{\As^3}/6 = 0$. Since $S_1$ has odd terms, it must be non-zero. By flipping the sign of $\rad_{ab}$, we can flip the sign of $S_1$, and so $-S_1 + S_2 \neq 0$. Therefore, for every configuration the variables $\rad_{ij}$ that result in a 0 trace, there exists an equally likely configuration that results in $\trace{\As^3}/6 \neq 0$. Therefore, $\Prob{\abs{\trace{\As^3}} > 0} \geq 1/2$.
            \item Every edge in $G$ is a part of an even number of triangles. Let $\curs{a,b,c}$ be a triangle in $G$. In this case, we decompose \eqref{eq:sum_over_all_triangles} as follows:
            \begin{align*}
                \trace{\As^3}/6 &= \underbrace{\rad_{ab} \rad_{bc} \rad_{ac}}_{S_1} +\underbrace{T(a,b) - \rad_{ab} \rad_{bc} \rad_{ac}}_{S_2} + \underbrace{T(b,c) - \rad_{ab} \rad_{bc} \rad_{ac}}_{S_3}  \\ &+\underbrace{T - T(a,b) - T(b,c) + \rad_{ab} \rad_{bc} \rad_{ac}}_{S_4}
            \end{align*}
            Consider all 4 possible values of the pair of random variables $\curs{\rad_{ab}, \rad_{bc}}$. Since each $S_i$ has an odd number of terms, $S_i \neq 0$ for all $i$. We observe that it is not possible for all 4 equally likely configurations of $\curs{\rad_{ab}, \rad_{bc}}$ to result in $S_1 + S_2 + S_3 + S_4 = 0$, so at least one configuration must result in $T \neq 0$. Therefore,  $\Prob{\abs{\trace{\As^3}} > 0} \geq 1/4$.
        \end{enumerate}
    \end{enumerate}
\end{proof}

By following the proof of Theorem 15 from \cite{musco2017spectrum}, and replacing their use of $\A_G$ with a symmetric random signing $\As_G$, we obtain an improved randomized version of their result by leveraging the smaller spectral radius of $\As_G$: 
\begin{restatable}[Improved randomized version of Theorem 15 from \cite{musco2017spectrum}]{theorem}{thmFifteenSigned}\label{thm:thm15_signed}
Let $f:\R^+ \to \R^+$ be a function such that it can be expressed as $f(x) = \sum_{k=0}^\infty c_k (x - 1)^k$ where $\abs{c_k/c_3} \leq h^{k-3}$ for $k > 3$ and $x \in (0,2)$. Given an algorithm which takes as input a graph $G = (V, E)$ on $n$ nodes and, in $O\pars{n^\gamma \epsilon^{-c}}$ time, outputs an estimate $X \approx_{\epsilon_1/9} \mathcal{S}_f(\I - \delta \As_G)$ with $\delta$ and $\epsilon_1$ satisfying 
\begin{align*}
    &\delta = \min\curs{({\sqrt{n} \log (\alpha n)})^{-1}, ({10n^3 h \log (\alpha n)})^{-1}}\\
    &\epsilon_1 = \min\curs{1, \abs{{c_3 \delta^3}/{(c_0 n)}}, \abs{{c_3 \delta}/{(c_2 n^2)}}}
\end{align*} 
for some constant $\alpha>1$, we can produce an algorithm that solves the triangle detection problem in $O\pars{n^2 + n^\gamma \epsilon_1^{-c}}$ time whp.
\end{restatable}
\begin{proof}
For convenience, we write $\A = \A_G$. The proof closely follows the proof of Theorem 15 from \cite{musco2017spectrum}, but we replace $\A$ with its symmetrically random signed version that we denote $\As$. We present a full proof here for completeness.

First, we note that by lemma~\ref{lemma:concentration}, we have that $\norm{\As}_2 \leq \sqrt{n} \log{\alpha n}$ whp. for some constant $\alpha >1$. We define $\Bs = \I - \delta \As$ and consequently $\Bs$ is PSD whp. Now, using the definition of $f$, we have
    \begin{align*}
        \sum_{i=1}^n \sigma_i(\Bs) = \sum_{i=1}^n f(1 - \delta \lambda_i(\As)) = \sum_{i=1}^n \sum_{k=0}^\infty c_k (\delta \lambda_i(\As))^k = \sum_{k=0}^\infty c_k \delta^k \trace{\As^k}.
    \end{align*}
We analyze the tail of this power series. Specifically, we have 
\begin{align}\label{eq:thm15proof_powerseries_tail_bound}
    \abs{\sum_{k=4}^\infty  c_k \delta^k \trace{\As^k}} \leq \abs{c_3} \delta^3 \sum_{k=4}^\infty \abs{\trace{\As^k}} \delta^{k-3} \abs{\frac{c_k}{c_3}}.
\end{align}
Now, we have $\abs{\trace{\As^k}} \leq \norm{\As}^{k-2}_2 \normgen{\As}{F} \leq n^{k/2 + 1}$ whp. and further $n^{k/2 + 1} \leq n^{3\pars{k-3}}$ for all $k>3$. Therefore, using the definition of $\delta$ as in the theorem we get that whp.,
    \begin{align*}
        \abs{\trace{\As^k}} \delta^{k-3} \abs{\frac{c_k}{c_3}} \leq \frac{1}{10^{k-3}}\text{ for all } k>3.
    \end{align*}
Plugging into \Cref{eq:thm15proof_powerseries_tail_bound}, we get 
\begin{align*}
    \abs{\sum_{k=4}^\infty  c_k \delta^k \trace{\As^k}} \leq \frac{\abs{c_3}\delta^3}{9}.
\end{align*}
The rest of the proof is essentially identical to the steps in the proof of Theorem 15 in \cite{musco2017spectrum}, but we reproduce them here for completeness.

Using the simple facts $\trace{\As^0} = n$, ${\trace{\As}} = 0$ and ${\trace{\As^2}} \leq n^2$, we have
\begin{align*}
    c_0\trace{\As^0} + c_1\trace{\As} + c_2 {\trace{\As^2}} \leq \abs{c_3}\delta^3 \pars{c_0 n/(c_3\delta^3) + c_2 n^2/(c_3 \delta)} \leq \frac{\abs{c_3} \delta^3}{\epsilon_1}.
\end{align*}
Given $X$, in $O(\nnz(\As))$ time, we can compute 
\begin{align*}
    X - c_0 n - c_2 \delta^2 {\trace{\As^2}} &= {c_3}\delta^3 \trace{\As^3} \pm \frac{\abs{c_3} \delta^3}{9} \pm \frac{\epsilon_1}{9}\pars{\frac{\abs{c_3}\delta^3}{9} + c_3 \delta^3 \trace{\As^3} +\frac{\abs{c_3}\delta^3}{\epsilon_1}}\\
    &= c_3 \delta^3 \sqrs{\trace{\As^3}\pars{1 \pm \frac{1}{20}} \pm \frac{1}{3}}.
\end{align*}
This is sufficient to detect if $\abs{\trace{\As}} = 0$ or if $\abs{\trace{\As^3}} \geq 1$. The final result then follows by applying Lemma~\ref{lemma:random_signing_preserves_triangles}.   
\end{proof}

We now prove \Cref{thm:our_results_spectralsums_lowerbound} by applying this result to the functions $f$ that define the corresponding spectral sums.

\begin{proof}[Proof of \Cref{thm:our_results_spectralsums_lowerbound}]
    We apply \Cref{thm:thm15_signed} to the specific spectral sums in \Cref{table:spectral_sum_hardness}.
    \paragraph{Schatten 3-norm} We have $f(x) = x^3$. Therefore, $c_k = 0$ for $k > 3$. So we apply Theorem \ref{thm:thm15_signed} with $h = 0$ and hence $\delta = \Otilde{\frac{1}{\sqrt{n}}}$ and $\epsilon_1 = \Otilde{\frac{1}{n^{2.5}}}$.
    
    \paragraph{Schatten p-norm $p \neq 1, 2$} We have $f(x) = x^p$. Using the Taylor series about 1, we have $\frac{c_k}{c_3} \leq p^{k-3}$ for all $k>3$ as well as $\abs{\frac{c_0}{c_3}} = \abs{\frac{1}{p(p-1)(p-2)}} \leq \abs{\frac{1}{2 \min \curs{p, (p-1), (p-2)}}}$ and similarly $\abs{\frac{c_2}{c_3}} \leq \abs{\frac{1}{2 \min \curs{p, (p-1)}}}$. Therefore, with $h=p$, we apply Theorem \ref{thm:thm15_signed} with $\delta = \Otilde{\frac{1}{n^3 p}}$ and $\epsilon_1 = \frac{c_3 \delta^3}{c_0 n} = \Otilde{\frac{\abs{\min \curs{p, (p-1), (p-2)}}}{n^{10}p^3}}$, which gives the result.

    \paragraph{SVD Entropy} We have $f(x) = x \log x$. For $x \in (0, 2)$, using the Taylor Series about 1 we can write $x \log x =\sum_{k=0}^\infty c_k(x-1)^k$ where $c_0 = 1 \log(1) = 0$, $c_1 = \log(1) + 1 = 1$, and $|c_k| = \frac{(k-2)!}{k!} \leq 1$ for $k \geq 2$. So we have $c_k < c_3$ for all $k > 3$, $\frac{c_0}{c_3} = 0$ and $\frac{c_2}{c_3} = \frac{1}{3}$. So with $h=1$,  Applying Theorem \ref{thm:thm15_signed} with $\delta = \Otilde{\frac{1}{n^3}}$ and $\epsilon_1 = \frac{\delta}{3n^2} = \Otilde{\frac{1}{n^5}}$ gives the result.

    \paragraph{Log Determinant} We have $f(x) = \log x$. For $x \in (0, 2)$, using the Taylor Series about 1 we can write $\log x =\sum_{k=0}^\infty c_k(x-1)^k$ where $c_0 = 0$, $\abs{c_i} = 1/i$ for $i \geq 1$. Again we have $c_k < c_3$ for all $k > 3$ and $\frac{c_0}{c_3} = 0$ while $\frac{c_2}{c_3} = \frac{3}{2}$. So with $h=1$,  Applying Theorem \ref{thm:thm15_signed} with $\delta = \Otilde{\frac{1}{n^3}}$ and $\epsilon_1 = \frac{\delta}{3n^2} = \Otilde{\frac{1}{n^5}}$ gives the result.

    \paragraph{Trace of Exponential} We have $f(x) = e^x$. Using the Taylor Series about $1$ we can write $e^x = \sum_{k=0}^\infty \frac{e(x-1)^k}{k!}$. We have $\frac{c_0}{c_3} = 6$, $\frac{c_2}{c_3} = 3$, and $c_k < c_3$ for all $k \geq 3$. So with $h=1$, Applying Theorem \ref{thm:thm15_signed} with $\delta = \Otilde{\frac{1}{n^3}}$ and $\epsilon_1 = \frac{c_3 \delta^3}{c_0 n} = \Otilde{\frac{1}{n^{10}}}$ gives the result.
\end{proof}
\section{Conclusion}
\label{sec:conclude}

In this paper we provided improved upper and lower bounds on the problem of estimating and sketching effective resistances on expanders. On the algorithmic side we show how sketches tailored to $\ell_1$ when carefully applied to asymmetric formulations of the quadratic form of the Laplacian pseudoinverse gave our results. On the lower bound side, we provided an alternative to the trace estimation approach of \cite{musco2017spectrum} for showing lower bounds and coupled it with techniques of randomly signing edges of the graph to obtain our results. Further, we showed that these techniques had broader implications for addressing algorithmic challenges in numerical linear algebra. 

Beyond the natural open problem of improving both our upper and lower bounds towards bringing them together there are interesting open problems in broadening the applicability of both our upper and lower bounds. For example, obtaining an $\Otilde{m \epsilon^{-1}}$ time algorithm for estimating the effective resistance of all edges in a general (non-expander) graph and extending our $\widetilde{\Omega}(n^2 \epsilon^{-1/2})$ lower bounds to deterministic algorithms remain interesting open problems. We hope that the results of this paper provide useful tools for addressing each.

\paragraph{Acknowledgements}
We thank Hongyue Li for helpful discussions and work on this project at various stages. Aaron Sidford was supported in part by a Microsoft Research Faculty Fellowship, NSF CAREER Award CCF-1844855, NSF Grant CCF-1955039, a PayPal research award, and a Sloan Research Fellowship. 

\bibliographystyle{unsrt}
\bibliography{ref}

\end{document}